\documentclass[oneside,reqno]{amsart}

\usepackage{amssymb}
\usepackage{overpic}
\usepackage{graphics}
\usepackage{epsfig}
\usepackage{amsmath}
\usepackage{amsfonts}
\usepackage{paralist}
\usepackage{setspace}
\usepackage{amsthm}
\usepackage{cite}

\newtheorem{theorem}{Theorem}[section]

\newtheorem{lemma}[theorem]{Lemma}
\newtheorem{proposition}{Proposition}

\newtheorem{example}{Example}
\newtheorem{property}{Property}
\theoremstyle{definition}
\newtheorem{definition}[theorem]{Definition}

\begin{document}


\title[Self-Avoiding Modes of Motion in a Lattice Gas]{Self-Avoiding Modes of Motion in a Deterministic Lorentz Lattice Gas}

\author{B. Z. Webb$^1$}

\author{E. G. D. Cohen$^2$}


\maketitle

\begin{center}
$^{1,2}$ Rockefeller University, Laboratory of Statistical Physics, 1230 York Avenue, New York, NY 10065, USA\\
E-mail: $^1$bwebb@rockefeller.edu and $^2$egdc@rockefeller.edu
\end{center}

\begin{spacing}{1}

\begin{abstract}
We study the motion of a particle on the two-dimensional honeycomb lattice, whose sites are occupied by either flipping rotators or flipping mirrors, which scatter the particle according to a deterministic rule. For both types of scatterers we find a new type of motion that has not been observed in a Lorentz Lattice gas, where the particle's trajectory is a self-avoiding walk between returns to its initial position. We show that this behavior is a consequence of the deterministic scattering rule and the particular class of initial scatterer configurations we consider. Since self-avoiding walks are one of the main tools used to model the growth of crystals and polymers, the particle's motion in this class of systems is potentially important for the study of these processes.
\end{abstract}

\section{Introduction}
In a Lorentz\footnote[1]{The \emph{Lorentz} in Lorentz lattice gas refers to H. A. Lorentz who assumed that the electrons in a conductor move independently of each other so that it is sufficient to study a single electron.} lattice gas (LLG) a single particle moves along the bonds of a lattice from lattice site to lattice site. When the particle arrives at a lattice site, it encounters a scatterer that modifies its motion according to a given scattering rule. Depending on the scattering rule, each scatterer can also have one of a number of orientations. Moreover, the orientation of each scatterer can be fixed or may change depending on this rule.

In an LLG, the initial orientation of each scatterer is called the LLG's initial configuration of scatterers. The trajectory of a particle in an LLG is then determined by the particular choice of (i) the lattice, (ii) the scattering rule, and (iii) the initial configuration of scatters on the lattice. In previous studies, a wide variety of dynamics has been observed in such systems, depending on the choice of these three features \cite{Bunimovich93,Cao97,Grosfils99,Ruijgrok88,Wang94,Wang95.1,Wang95.3,Wang95.2}.

In this paper, we consider the particle's motion on the regular two-dimensional honeycomb lattice, in which the lattice is either fully occupied by flipping rotators or by flipping mirrors. A flipping rotator and a flipping mirror have one of two orientations: either \emph{left} or \emph{right}. On the honeycomb lattice, a flipping rotator scatters the particle by rotating its velocity by the angle $\theta=\pm\pi/3$, either to its left or to its right, depending on whether the rotator is a left or a right rotator, respectively. Similarly, a flipping mirror reflects the particle's velocity over the angle $\theta=\pm\pi/3$, to its left or right, depending on whether it is a left or a right mirror, respectively. In both cases, the scatterers change or \emph{flip} orientation after scattering the particle, from either right to left or from left to right.

The initial configuration we consider, throughout the majority of the paper, is the configuration in which each scatterer is a right rotator. Despite the simplicity of this initial configuration, we observe that the particle in this LLG has an unusual mode of motion. In this LLG the particle returns to its initial position an infinite number of times where, between these returns, the particle's trajectory is a self-avoiding walk (see theorem \ref{thm:1}). To the best of our knowledge, this is the first time this type of behavior has been observed in an LLG.

One of our main objectives in this paper is to describe the underlying mechanism that causes the particle to exhibit this type of motion. What we find is that the initial configuration of all right rotators has two special properties that lead to this behavior. The first is that it is an \emph{admissible} configuration (see definition \ref{def:admiss}), the second is that it is also a \emph{blocking} configuration (see definition \ref{def:blocking}).

These two new concepts, of an admissible configuration and a blocking configuration, allow us to prove the following general results. If the particle moves on the honeycomb lattice with an initial configuration that is admissible then, between returns to its initial position, the particle's trajectory is a \emph{self-avoiding} walk (see theorem \ref{lem:1}). If a particle moves on a lattice with an initial configuration that is both an admissible and a blocking configuration, not only does the particle have this self-avoiding property, it also returns to its initial position an infinite number of times (see theorem \ref{thm:recurrence}).

This self-avoiding property, which we find in the class of LLGs with an admissible configuration, is novel for a number of reasons. One is that the particle's motion in these systems is completely deterministic, i.e. this motion arises from the particle's deterministic equations of motion. This is in contrast to the large majority of self-avoiding walks, which are generated via some random process \cite{Amit83,Madras88}. A second reason is that a number of basic mathematical questions regarding self-avoiding walks \cite{Madras13} are still unresolved. Therefore, the rigorous results presented in this paper are potentially important for gaining a better mathematical understanding of self-avoiding walks. Such results also extend the mathematical theory of LLGs (see, for instance, \cite{Bunimovich93,Grosfils99,Bunimovich91}).

In addition, self-avoiding walks are one of the main tools used to model the growth of crystals and polymers \cite{Madras13,Bous92}. These walks also play a central role in the study of the folding and knotting behaviour of biological molecules such as proteins \cite{Guttmann12}. Therefore, the class of LLGs we consider in this paper is potentially important for the study of these processes.

Besides exhibiting this new type of motion, these LLGs are different from those that have been previously studied, in that they have a single initial configuration. In previous studies, a typical LLG has a large number of initial configurations, which makes it possible to characterize the particle's dynamics in terms of its mean square displacement \cite{Cao97,Grosfils99,Ruijgrok88,Wang94,Wang95.1,Wang95.3,Wang95.2,Kong89,Kong90.2}.
In the LLGs we consider, here this is not possible, because of the irregularity of the particle's mean square displacement.

In order to describe the particle's dynamics in the LLGs we consider, we introduce the notion of a \emph{time-averaged mean square displacement} to the study of these systems. This concept allows us to describe the particle's dynamics in the LLGs we consider in the same way that the mean square displacement is used to describe the particle's dynamics in previous studies. What we observe, numerically, for those LLGs that have a configuration that is both admissible and blocking, is that the particle's time-averaged mean square displacement increases as a power law. This suggests that the particle in any one of these systems has what we refer to as a \emph{pulsating} motion since it regularly returns to its initial position but also appears to have an unbounded trajectory.

The paper is organized as follows. In section \ref{sec:3} we describe the flipping rotator and flipping mirror systems we consider throughout the paper, along with their respective equations of motion. We show that for each flipping rotator system there is a corresponding flipping mirror system with the same dynamics, and vice versa. In section \ref{sec:4} we investigate the flipping rotator system, in which all scatterers are initially right rotators. In section \ref{sec:tamsd} we introduce the concept of  a particle's time-averaged mean square displacement and the idea of a pulsating motion of a particle in an LLG.

In section \ref{sec:5} we introduce the class of admissible configurations, that lead to the particle's self-avoiding trajectory, between returns to its initial position. In section \ref{sec:7} we introduce the class of blocking configurations, and show that configurations that are both admissible and blocking lead the particle to return to its initial position an infinite number of times, i.e. have a pulsating-like behavior. In section \ref{sec:conc} we give a number of concluding remarks.

We note that although the main results of this paper are proven mathematically, the paper is written so that it can be understood without the need for the reader to work through the various proofs. To facilitate this, each of the paper's mathematical results is preceded by an explanation of the result in words.

\begin{figure}
    \begin{overpic}[scale=.3]{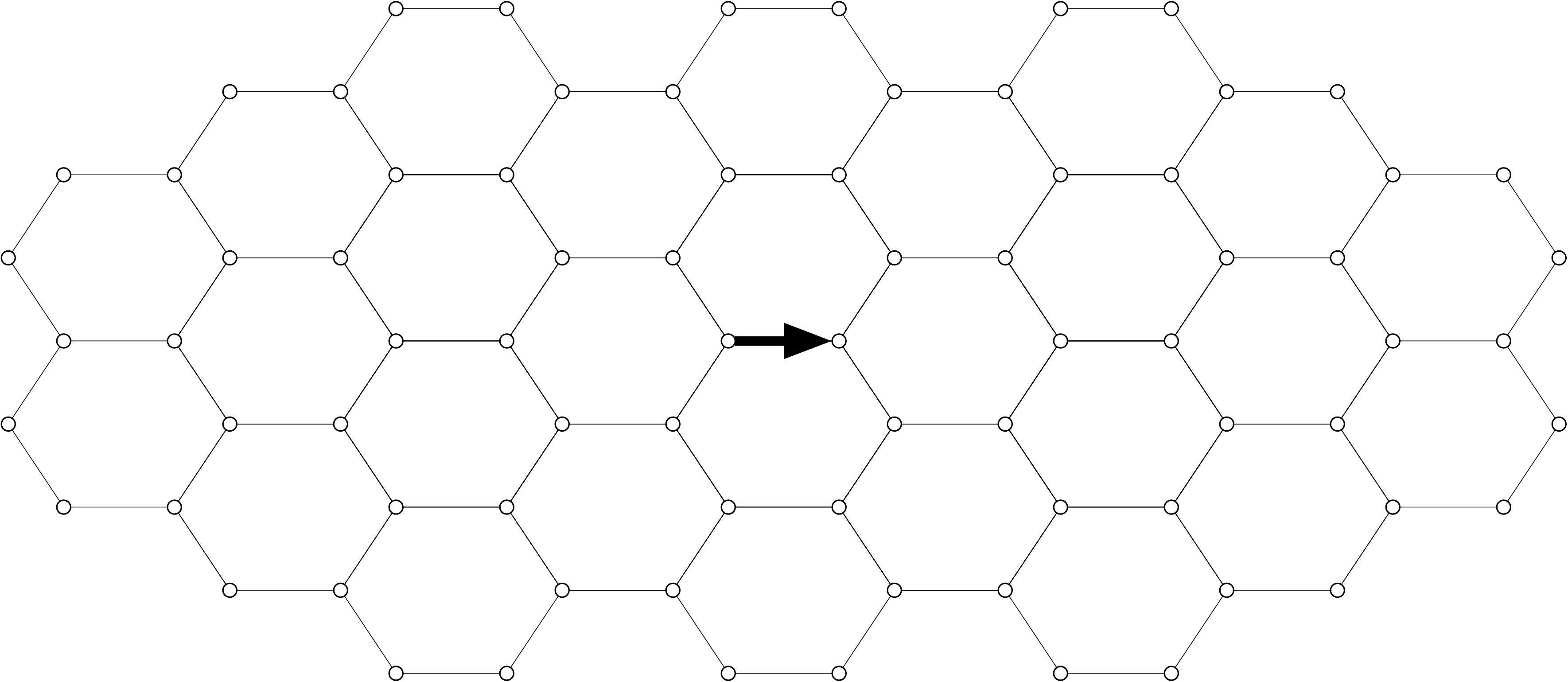}
    \put(43,20.75){$\mathbf{r}$}
    \put(49,24){$\mathbf{v}$}
    \end{overpic}
\caption{The regular honeycomb lattice $H=(\mathbb{H},\mathbb{B})$ where the particle's initial position $\mathbf{r}=(0,0)$ and velocity $\mathbf{v}=(1,0)$ are indicated.}\label{fig:1}
\end{figure}

\section{The Flipping Rotator and Flipping Mirror Systems}\label{sec:3}
In this section we begin by describing the flipping rotator and flipping mirror systems considered in this paper. We then show that the two types of systems are dynamically equivalent, in the sense that to every flipping rotator system there is a corresponding flipping mirror system with the same dynamics and that the converse of this also holds.

The lattice we consider in this paper is the honeycomb lattice $H=(\mathbb{H},\mathbb{B})$, with sites $\mathbb{H}$ and bonds $\mathbb{B}$. The lattice consists of regular hexagons, with sides of unit length. Hence, there is a lattice bond of length $1$ between two sites of $H$ if and only if these sites are nearest neighbors (cf. figure \ref{fig:1}).

We let $\mathbf{r}(t)\in\mathbb{R}^2$ denote the \emph{position} and $\mathbf{v}(t)\in\mathbb{R}^2$ denote the \emph{velocity} of the particle at time $t\geq 0$ on the lattice, where the particle is assumed to move with constant unit speed. Moreover, we let $I=(\mathbf{r},\mathbf{v})$ denote the particle's \emph{initial state}, i.e. initial position and initial velocity, where $\mathbf{r}=\mathbf{r}(0)$ and $\mathbf{v}=\mathbf{v}(0)$.

We restrict ourselves to \emph{discrete} time steps $t=0,1,2,\dots$, so that the particle is at some lattice sight for each time $t$. The particle's \emph{trajectory} is then the semi-infinite sequence of positions $\{\mathbf{r}(t)\}_{t\geq 0}\subset\mathbb{H}$. Since the velocity of the particle does not exist at the moment it is scattered, we let $\mathbf{v}(t)$ denote the velocity of the particle directly after each time $t\geq 0$.

We assume that there is a scatterer at each lattice site of $H$, so that the lattice is either fully occupied by rotators or mirrors. We note that each scatterer is initially either a right scatterer or a left scatterer. We let $C=C(0)$ denote this \emph{initial configuration} of scatterers and let $C(t)$ denote the configuration of scatterers on the lattice at time $t\geq 0$. For each lattice site $\mathbf{h}\in\mathbb{H}$ we let
\begin{equation}
C_{\mathbf{h}}(t)\in\{-1,1\} \ \ \text{for} \ \mathbf{h}\in\mathbb{H}, \  t\geq 0
\end{equation}
denote the orientation of the scatterer, either left or right, at site $\mathbf{h}$ at time $t$. The orientation $C_{\mathbf{h}}(t)=-1$ indicates that at time $t$ the scatterer at $\mathbf{h}$ is a left scatterer whereas the orientation $C_{\mathbf{h}}(t)=1$ indicates that the scatterer is a right scatterer.

\begin{figure}
    \begin{overpic}[scale=.5]{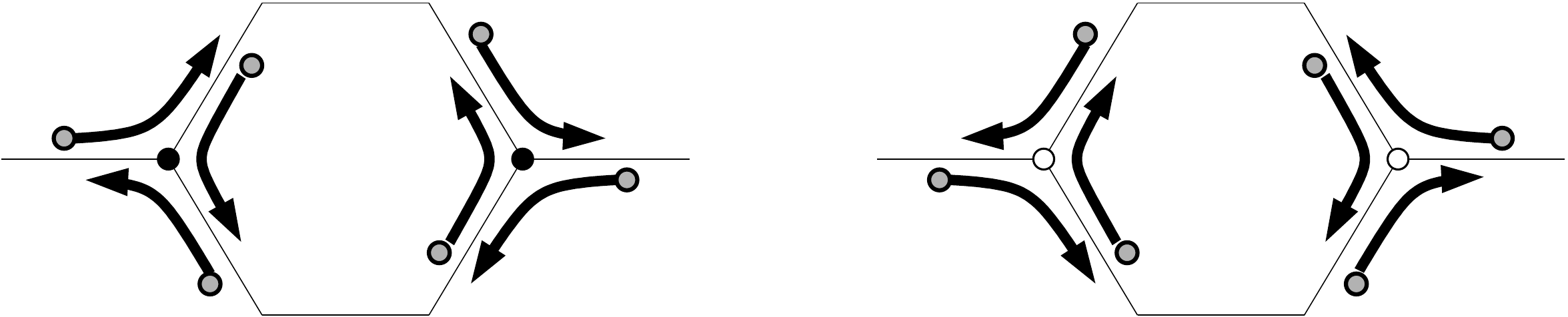}
    \put(8,-4.5){(a) \small\emph{left rotators} $(\mathbf{LR})$}
    \put(64,-4.5){(b) \small\emph{right rotators} $(\mathbf{RR})$}
    \end{overpic}
    \vspace{0.2in}
\caption{Upon arriving at a left (right) rotator the particle's velocity is rotated to its left (right) by an angle of $\theta=\pm\pi/3$. In this and following figures, closed circles denote left rotators and open circles denote right rotators. Solid lines denote lattice bonds.}\label{fig:2}
\end{figure}

In the case of flipping rotators, each scatterer rotates the velocity of the incoming particle either to its right or its left by an angle of $\theta=\pm\pi/3$. This is shown in figure \ref{fig:2}, where we use the convention that a closed circle denotes a left rotator $(\mathbf{LR})$ and an open circle denotes a right rotator $(\mathbf{RR})$, respectively.

The equations of motion for flipping rotators are given by
\begin{align}
\mathbf{r}(t+1)&=\mathbf{r}(t)+\mathbf{v}(t), \label{eq:1}\\
\mathbf{v}(t+1)&=R\big[C_{\mathbf{h}}(t)\big]\mathbf{v}(t) \ \ \text{for} \ \ \mathbf{r}(t+1)=\mathbf{h}, \label{eq:2}\\
C_{\mathbf{h}}(t+1)&=
\begin{cases}
-C_{\mathbf{h}}(t) &  \ \ \text{if} \ \  \mathbf{r}(t+1)=\mathbf{h}\\
\hspace{0.1in} C_{\mathbf{h}}(t) & \ \ \text{if} \ \ \mathbf{r}(t+1)\neq\mathbf{h}\label{eq:3}
\end{cases}
\end{align}
for $t\geq 0$. Equation \eqref{eq:1} gives the dynamics of the particle, describing its piecewise linear motion between successive scatterings. In equation (\ref{eq:2}), the rotation operator $R:\{-1,1\}\rightarrow\mathbb{R}^{2\times 2}$ is the matrix given by
\begin{equation}\label{eq:matrix}
R(z)=
\left[\begin{array}{cc}
\cos(\frac{\pi}{3}z)&\sin(\frac{\pi}{3}z)\\
-\sin(\frac{\pi}{3}z)&\cos(\frac{\pi}{3}z)
\end{array}\right],
\end{equation}
which describes how the particle's velocity is rotated when it arrives at a scatterer. Equation \eqref{eq:3} describes the flipping motion of the scatterers.

\begin{definition}\label{flip}
Let $(H_{fr},I,C)$ denote the LLG on the honeycomb lattice $H$ with initial state $I$, initial configuration $C$, and equations of motion given by equations $\eqref{eq:1}-\eqref{eq:3}$. We call this a \emph{flipping rotator} (\emph{fr}) \emph{system} on the honeycomb lattice.
\end{definition}

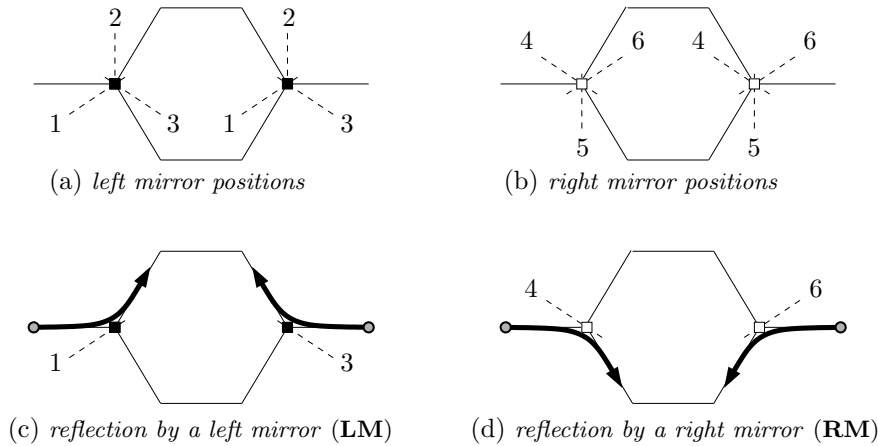
\begin{figure}
    \begin{overpic}[scale=.425]{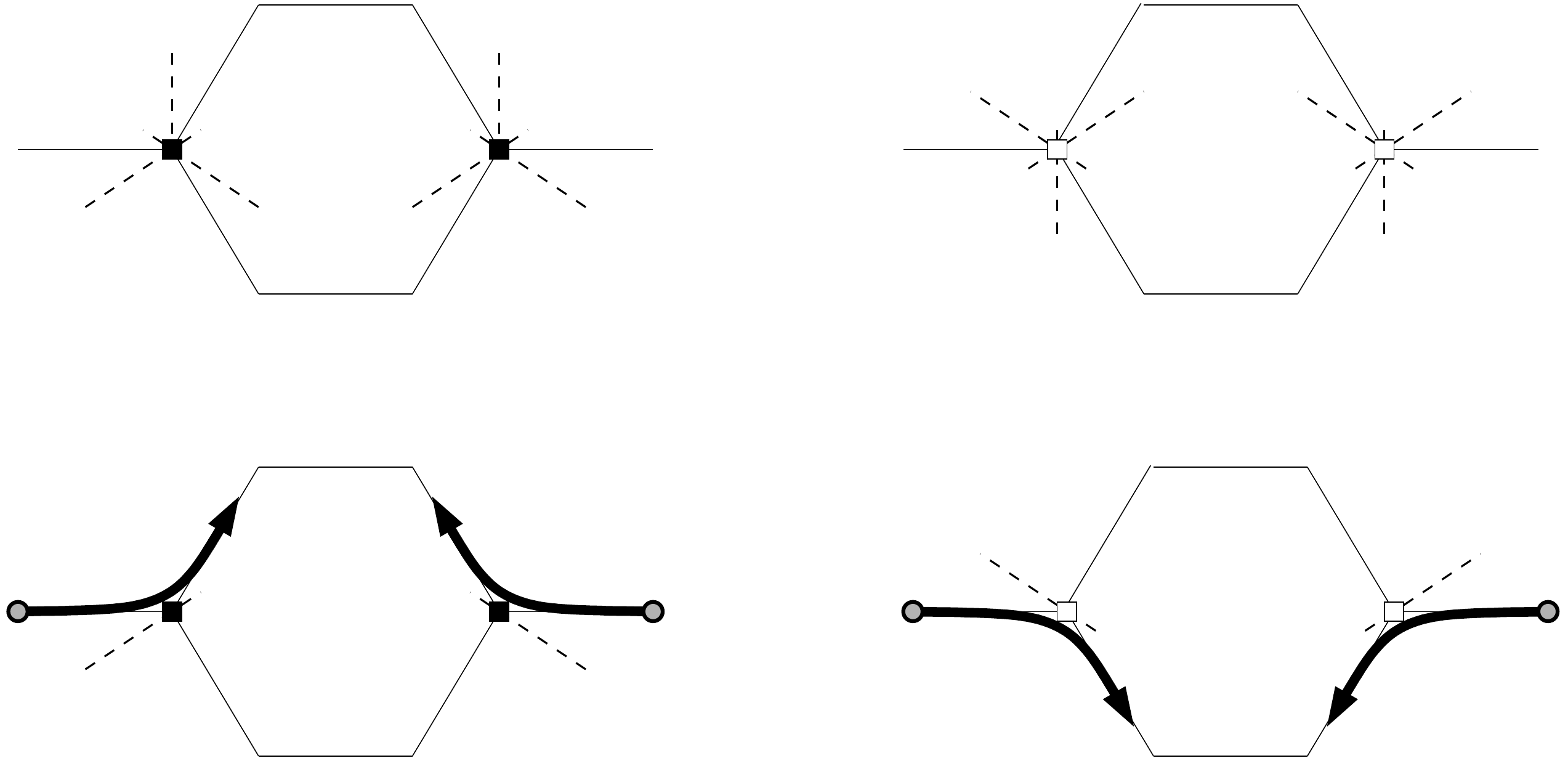}
    \put(3,26){(a) \small\emph{left mirror positions}}
    \put(58,26){(b) \small\emph{right mirror positions}}
    \put(-2,-4){(c) \small\emph{reflection by a left mirror} $(\mathbf{LM})$}
    \put(54,-4){(d) \small\emph{reflection by a right mirror} $(\mathbf{RM})$}

    \put(3,33){$1$}
    \put(10.25,45.75){$2$}
    \put(17.25,33){$3$}

    \put(24,33){$1$}
    \put(31.25,45.75){$2$}
    \put(38.25,33){$3$}

    \put(3,4){$1$}
    \put(38.25,4){$3$}

    \put(60,43){$4$}
    \put(66.75,30){$5$}
    \put(73.5,43){$6$}

    \put(80.75,43){$4$}
    \put(87.5,30){$5$}
    \put(94.25,43){$6$}

    \put(60.5,13){$4$}
    \put(95,13){$6$}

    \end{overpic}
    \vspace{0.2in}
\caption{The three possible positions a left and right mirror can occupy are shown as dashed lines in (a) and (b), respectively. When the particle approaches a lattice site, the mirror at this site moves to the position nearest to the incoming particle. This is shown in (c) and (d) for left and right mirrors, respectively. Closed squares denote left mirrors and open squares denote right mirrors.}\label{fig:2.01}
\end{figure}

In the case of flipping mirrors each scatterer is a single mirror, which reflects the particle over an angle $\theta=\pm\pi/3$, either to its left or right. Each mirror is double-sided and is based at a lattice site. The mirrors have two orientations, either left or right. A left mirror $(\mathbf{LM})$ can\emph{move} to any one of the three positions 1-3 indicated by the dashed lines in figure \ref{fig:2.01} (a), whereas a right mirror $(\mathbf{RM})$ can move to any one of the three positions 4-6 indicated in \ref{fig:2.01} (b). When the particle moves along a lattice bond towards a lattice site, the mirror at this site moves to the position nearest to the incoming particle. Once the mirror has reflected the particle it then flips its orientation from either right to left or from left to right.\footnote[2]{To physically implement this, one could put an optical device at each lattice site of the lattice that realizes a flipping mirror. The device should then be able to see the incoming particle and would contain an adjustable mirror, which it would position to scatter the particle.}

The way in which the particle is scattered by either a left or right mirror is demonstrated in figure \ref{fig:2.01} (c) and (d) respectively. In figure \ref{fig:2.01} (c) the mirror on the left moves to position 1 since this is the position closest to the particle as it enters from the left. Similarly, the mirror on the right moves to position 3 since this is the position closest to the particle as it enters from the right. In figure \ref{fig:2.01} (d) the right mirrors moves to the position 4 and 6 since these are positions closest to the incoming particles, respectively.

The geometric reason we restrict the mirrors to the six positions, indicated in figure \ref{fig:2.01}, is that a mirror in any other position at a lattice site would reflect the particle off the lattice. The lattice geometry is also the reason for these six positions. Since the particle can approach a lattice site along three different lattice bonds, there must be three different positions the mirror can be in so as to reflect the particle to either to its left or right, for each orientation of the mirror.

To describe the relationship between rotators and mirrors, we first note the following. Each lattice site $\mathbf{h}\in\mathbb{H}$, on the honeycomb lattice, is at the end of exactly one horizontal lattice bond (cf. figure \ref{fig:1}). Using this, we say that $\mathbf{h}\in\mathbb{H}^+$, if $\mathbf{h}$ is on the right-hand side of this bond and $\mathbf{h}\in\mathbb{H}^-$, if $\mathbf{h}$ is on the left-hand side of the bond. That is, we define the sets
\begin{align*}
\mathbb{H}^+&=\{\mathbf{h}\in\mathbb{H}: \ \exists \ \mathbf{h}_{\ell}\in\mathbb{H},\mathbf{h}=\mathbf{h}_{\ell}+(1,0)\}\\
\mathbb{H}^-&=\{\mathbf{h}\in\mathbb{H}: \ \exists \ \mathbf{h}_{r}\in\mathbb{H},\mathbf{h}=\mathbf{h}_{r}-(1,0)\}.
\end{align*}

As one can check, $\mathbb{H}=\mathbb{H}^+\cup\mathbb{H}^-$ and $\mathbb{H}^+\cap\mathbb{H}^-=\emptyset$, i.e. each lattice site is either on the left or on the right end of a horizontal lattice bond. Moreover, if $\mathbf{h}\in\mathbb{H}^+$ then an $\mathbf{RM}$ at $\mathbf{h}$ acts as an $\mathbf{RR}$ and an $\mathbf{LM}$ at $\mathbf{h}$ acts as an $\mathbf{LR}$. If $\mathbf{h}\in\mathbb{H}^-$ then an $\mathbf{RM}$ at $\mathbf{h}$ acts as an $\mathbf{LR}$ and an $\mathbf{LM}$ at $\mathbf{h}$ acts as an $\mathbf{RR}$. This is summarized in table \ref{table:0} and is demonstrated in figure \ref{fig:2.1}.

\begin{table}[ht]

\centering 

\begin{tabular}{| c | c | c | c |}
\hline
   $\mathbf{LM}^+\equiv\mathbf{LR}$ & $\mathbf{RM}^+\equiv\mathbf{RR}$ & $\mathbf{LM}^-\equiv\mathbf{RR}$ & $\mathbf{RM}^-\equiv\mathbf{LR}$\\
\hline
\end{tabular}
\vspace{0.15in}
\caption{The table shows the relation between rotators and mirrors at lattice sites of $\mathbb{H}^+$ and $\mathbb{H}^-$. The notation $\mathbf{LM}^+\equiv\mathbf{LR}$ indicates, for instance, that a left mirror acts as a left rotator at a lattice site of $\mathbb{H}^+$.} 
\label{table:0} 
\end{table}

\begin{figure}
    \begin{overpic}[scale=.5]{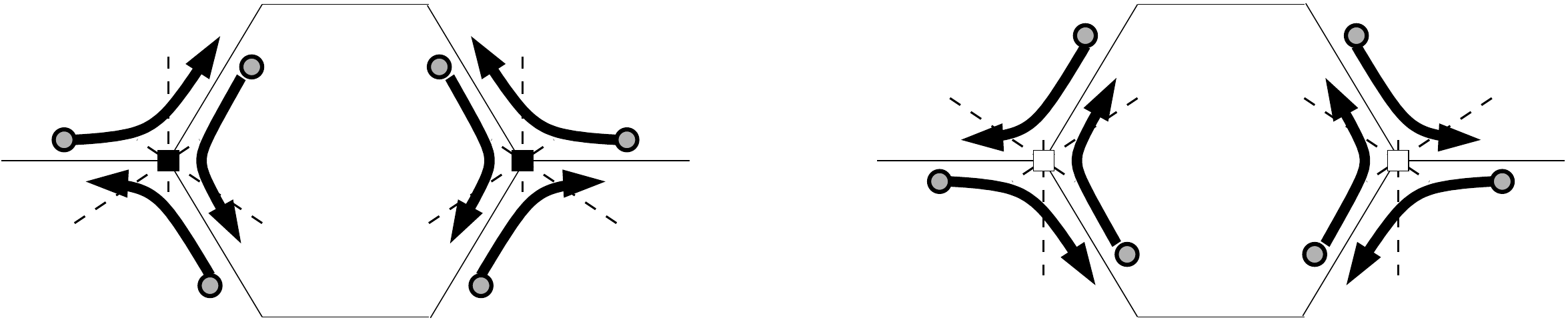}
    \put(3,4.5){$1$}
    \put(10,17.5){$2$}
    \put(17.5,4.5){$3$}

    \put(25.5,4.5){$1$}
    \put(32.5,17.5){$2$}
    \put(40,4.5){$3$}

    \put(58.5,14){$4$}
    \put(66,1){$5$}
    \put(73.5,14){$6$}

    \put(81,14){$4$}
    \put(88.5,1){$5$}
    \put(96,14){$6$}

    \put(27.5,9.33){$-$}
    \put(14,9.33){$+$}

    \put(83.5,9.33){$-$}
    \put(70,9.33){$+$}

    \put(7,-10){(a) \small\emph{left mirrors} $(\mathbf{LM})$}
    \put(5,-3){\small\emph{left rotator}}
    \put(25,-3){\small\emph{right rotator}}

    \put(66,-10){(b) \small\emph{right mirrors} $(\mathbf{RM})$}
    \put(59,-3){\small\emph{right rotator}}
    \put(83,-3){\small\emph{left rotator}}
    \end{overpic}
    \vspace{0.5in}
\caption{Depending on whether a mirror is at a lattice site of $\mathbb{H}^+$ or $\mathbb{H}^-$, it acts as either a left or right rotator. This is shown in (a) for left mirrors and in (b) for right mirrors, where $\pm$ indicates a lattice site of $\mathbb{H}^{\pm}$, respectively. This is summarized in table \ref{table:0}.}
\label{fig:2.1}
\end{figure}
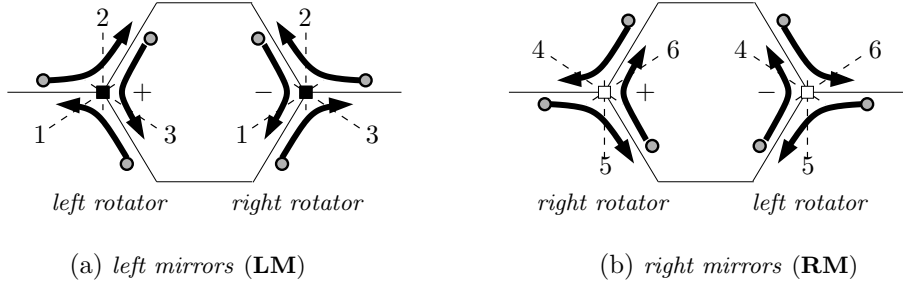

This relation between mirrors and rotators allows us to write the equations of motion for a flipping mirror system in terms of the rotation operator $R$, given in \eqref{eq:matrix}. These equations are given by
\begin{align}
\mathbf{r}(t+1)&=\mathbf{r}(t)+\mathbf{v}(t), \label{eq:4}\\
\mathbf{v}(t+1)&=
\begin{cases}
R\big[C_{\mathbf{h}}(t)\big]\mathbf{v}(t) \ \ &\text{if} \ \ \mathbf{r}(t+1)=\mathbf{h}\in\mathbb{H}^+\\
R\big[-C_{\mathbf{h}}(t)\big]\mathbf{v}(t) \ \ &\text{if} \ \ \mathbf{r}(t+1)=\mathbf{h}\in\mathbb{H}^- \label{eq:5}
\end{cases},\\
C_{\mathbf{h}}(t+1)&=
\begin{cases}
-C_{\mathbf{h}}(t) &  \ \ \text{if} \ \  \mathbf{r}(t+1)=\mathbf{h}\\
\hspace{0.1in} C_{\mathbf{h}}(t) & \ \ \text{if} \ \ \mathbf{r}(t+1)\neq\mathbf{h}\label{eq:6}
\end{cases}
\end{align}
for $t\geq 0$.

\begin{definition}\label{flipmirror}
Let $(H_{fm},I,C)$ denote the LLG on the honeycomb lattice $H$ with initial state $I$, initial configuration $C$, and equations of motion given by equations $\eqref{eq:4}-\eqref{eq:6}$. We call this a \emph{flipping mirror} (\emph{fm}) \emph{system} on the honeycomb lattice.
\end{definition}

The dynamics of flipping rotator and flipping mirror systems were considered separately in \cite{Wang95.1}, where diffusion processes on the honeycomb lattice were studied. Here, we want to show that the particle in a flipping rotator system and a flipping mirror system can have the same trajectory. To demonstrate this, let $C$ be an initial configuration of scatterers. We define $\varphi(C)$ to be the configuration given by
\begin{equation}\label{eq:map}
\varphi(C)_{\mathbf{h}}=
\begin{cases}
\hspace{0.1in}  C_{\mathbf{h}} &  \ \ \text{if} \ \ \mathbf{h}\in\mathbb{H}^+\\
-C_{\mathbf{h}} &  \ \ \text{if} \ \ \mathbf{h}\in\mathbb{H}^-
\end{cases}.
\end{equation}
Thus, $\varphi(C)$ is the configuration $C$ in which the scatterer's orientation at each lattice site of $\mathbb{H}^-$ has been flipped.

The following proposition states that the particle's trajectory in a flipping rotator system is the same as the particle's trajectory in a flipping mirror system if one system has the initial configuration $C$ and the other the initial configuration $\varphi(C)$.

\begin{proposition}\label{prop:1}
For any initial configuration $C$, the particle in the system $(H_{fr},I,C)$ and the system $(H_{fm},I,\varphi(C))$ have the same trajectory. Moreover, the particle in the system $(H_{fm},I,C)$ and the system $(H_{fr},I,\varphi(C))$ have the same trajectory.
\end{proposition}

\begin{proof} Suppose $C$ is an initial configuration on $H$. Let $\mathbf{r}(t)$ and $\mathbf{v}(t)$ denote the particle's position and velocity respectively, in the system $(H_{fr},I,C)$. Similarly, let $\mathbf{s}(t)$ and $\mathbf{w}(t)$ denote the particle's position and velocity respectively, in the system $(H_{fm},I,\Phi)$ where $\Phi=\varphi(C)$.

At time $t=0$ we note that $\mathbf{r}(0)=\mathbf{s}(0)$, $\mathbf{v}(0)=\mathbf{w}(0)$, and $\Phi(0)=\varphi(C)(0)$. Proceeding by induction suppose that at a fixed time $\tau\geq 0$ that $\mathbf{r}(\tau)=\mathbf{s}(\tau)$, $\mathbf{v}(\tau)=\mathbf{w}(\tau)$, and $\Phi(\tau)=\varphi(C)(\tau)$. Under this assumption, equations \eqref{eq:1} and \eqref{eq:4} imply that $\mathbf{r}(\tau+1)=\mathbf{s}(\tau+1)$. Also, from equation \eqref{eq:5} it follows that
\begin{equation}\label{eq:split}
\mathbf{w}(\tau+1)=
\begin{cases}
R\big(\Phi_{\mathbf{r}(\tau+1)}(\tau)\big)\mathbf{v}(\tau) &  \ \ \text{if} \ \ \mathbf{r}(\tau+1)\in\mathbb{H}^+\\
R\big(-\Phi_{\mathbf{r}(\tau+1)}(\tau)\big)\mathbf{v}(\tau) &  \ \ \text{if} \ \ \mathbf{r}(\tau+1)\in\mathbb{H}^-
\end{cases}.
\end{equation}
Since $\Phi_{\mathbf{r}(\tau+1)}(\tau)=C_{\mathbf{r}(\tau+1)}(\tau)$ if $\mathbf{r}(\tau+1)\in\mathbb{H}^+$ and $-\Phi_{\mathbf{r}(\tau+1)}(\tau)=C_{\mathbf{r}(\tau+1)}(\tau)$ if $\mathbf{r}(\tau+1)\in\mathbb{H}^-$ then equation \eqref{eq:2} together with equation \eqref{eq:split} imply that $\mathbf{v}(\tau+1)=\mathbf{w}(\tau+1)$. Moreover, using equations \eqref{eq:3} and \eqref{eq:6} we have that
\begin{equation}
\Phi(\tau+1)=
\begin{cases}
-\varphi(C)_{\mathbf{h}}(\tau) &  \ \ \text{if} \ \ \mathbf{r}(\tau+1)=\mathbf{h}\\
\hspace{0.1in} \varphi(C)_{\mathbf{h}}(\tau) &  \ \ \text{if} \ \ \mathbf{r}(\tau+1)\neq\mathbf{h}
\end{cases}=\varphi(C)(\tau+1).
\end{equation}

Since $\mathbf{r}(\tau+1)=\mathbf{s}(\tau+1)$, $\mathbf{v}(\tau+1)=\mathbf{w}(\tau+1)$, and $\Phi(\tau+1)=\varphi(C)(\tau+1)$ then, by induction, it follows that $\mathbf{r}(t)=\mathbf{s}(t)$ for all $t\geq 0$. Therefore, the particle's trajectory in the systems $(H_{fr},I,C)$ and $(H_{fm},I,\varphi(C))$ is the same for any initial configuration $C$. Using the initial configuration $\varphi(C)$ this implies that $(H_{fr},I,\varphi(C))$ and $(H_{fm},I,\varphi^2(C))$ have the same trajectory. Since $\varphi^2(C)=C$, by use of equation \eqref{eq:map}, this completes the proof.
\end{proof}

The main consequence of proposition \ref{prop:1} is that, any behavior that is observed in a flipping rotator system can also be observed in a corresponding flipping mirror system, and vice-versa. In this sense, the flipping rotator and flipping mirror systems are dynamically equivalent. This fact will allow us in the remainder of the paper, to concentrate our attention on systems with flipping rotators alone. Our results on such systems will, via proposition \ref{prop:1}, immediately imply the same results for the corresponding flipping mirror systems.

Because of this, we will from here on simplify our notation by letting $(H,I,C)$ denote the flipping rotator system $(H_{fr},I,C)$. Moreover, if the initial condition $I=(\mathbf{r},\mathbf{v})$ is given by $\mathbf{r}=(0,0)$ and $\mathbf{v}=(1,0)$ we suppress the system's dependence on its initial state and write $(H,I,C)$ as $(H,C)$, cf. figure \ref{fig:1}.

\section{The Initial Configuration of All Right Scatterers}\label{sec:4}
In this section we begin by investigating the flipping rotator system whose initial configuration is the configuration of all right scatterers. This system, which we denote by $(H,\mathbf{R})$, serves as our primary example of a system in which the particle's trajectory is a self-avoiding walk between returns to its initial position. Here, we describe how the particle's trajectory can be decomposed into an infinite number of cycles, each of which contains the origin. We then study the statistical properties of these cycles.

We first give a description of the motion of the particle in the $(H,\mathbf{R})$ system. To do so we introduce the following terminology.

\begin{definition}
Suppose in the system $(H,I,C)$ there are two times $t_1<t_2$, such that each position $\mathbf{r}(t)$ is distinct for $t_1 \leq t< t_2$ and $\mathbf{r}(t_1)=\mathbf{r}(t_2)$. Then the particle is said to move on the \emph{cycle}
$$\gamma=\{\mathbf{r}(t):t_1\leq t\leq t_2\}$$
from time $t_1$ to $t_2$. We call the position $\mathbf{r}(t_1)=\mathbf{r}(t_2)$ the \emph{base} of the cycle.
\end{definition}

The following theorem states that the particle's trajectory in the $(H,\mathbf{R})$ system can be described in terms of cycles\footnote[3]{We note that in some contexts, cycles are referred to as \emph{self-avoiding polygons} \cite{Guttmann12}.}. Specifically, the particle's entire trajectory can be decomposed into an infinite sequence of cycles each of which is based at the origin.

\begin{theorem}\label{thm:1} \textbf{(Cyclic Decomposition Theorem)}
In the $(H,\mathbf{R})$ system there is an infinite sequence of times $\{\tau_i\}_{i\geq 0}$ with $\tau_0=0$, such that for each $i>0$\\
(a) the particle moves on a cycle $\gamma_i$  based at the origin $\mathbf{r}=(0,0)$, given by $$\gamma_i=\big\{\mathbf{r}(t):\tau_{i-1}\leq t\leq \tau_{i}\big\}; \ \text{and}$$
(b) each cycle $\gamma_i$ is symmetric with respect to the vertical line $x=1/2$.
\end{theorem}

We save a proof of theorem \ref{thm:1} for section \ref{sec:7}. The reason we defer this proof is that we first need to understand how the particle modifies the configuration of all right scatterers, as it moves through the lattice. This is the main topic considered in section \ref{sec:5}. In section \ref{sec:7} we study the particle's infinite sequence of returns to the origin in the $(H,\mathbf{R})$ system. Only after we have proved a number of intermediate results in these two sections, will we be able to give a proof of theorem \ref{thm:1}.

\begin{figure}
\begin{center}
\begin{tabular}{ccc}
    \begin{overpic}[scale=.38]{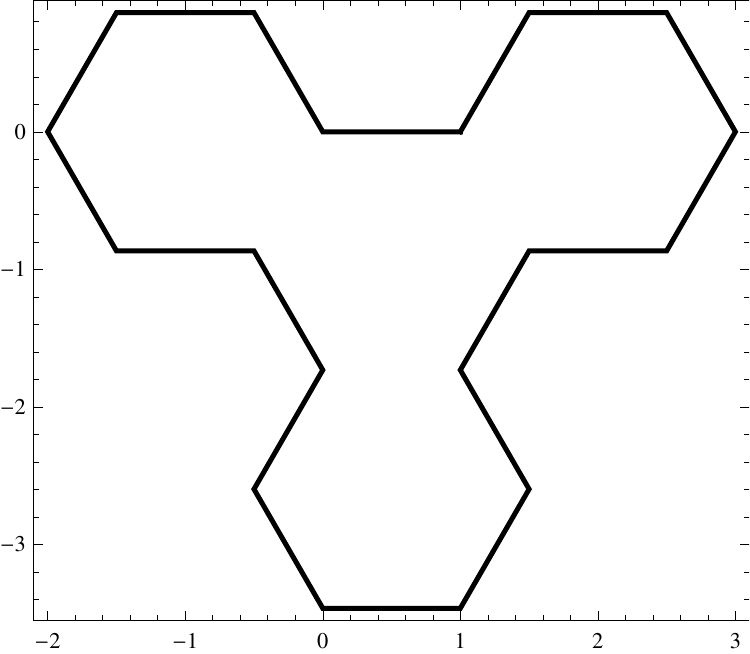}
    \put(47,-7){$\gamma_2$}
    \put(40,67){$\bullet$}
    \end{overpic} &
    \begin{overpic}[scale=.38]{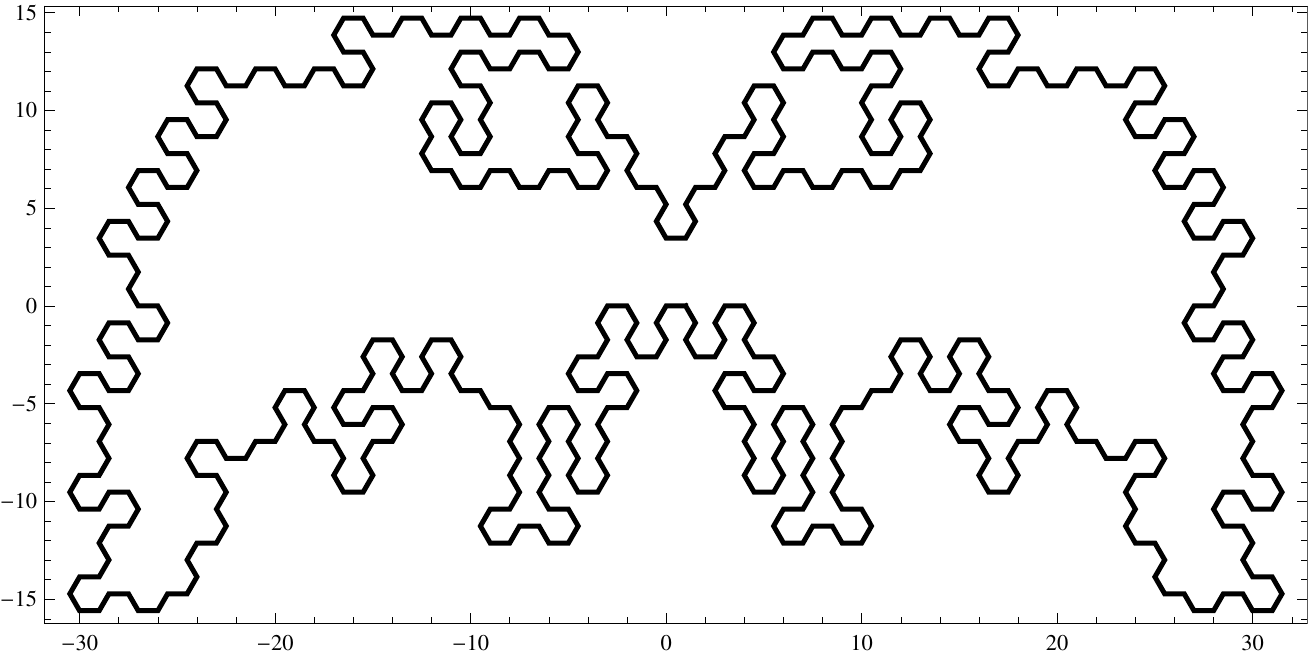}
    \put(45,-4){$\gamma_{819}$}
    \put(49,25){$\bullet$}
    \end{overpic} &
    \begin{overpic}[scale=.38]{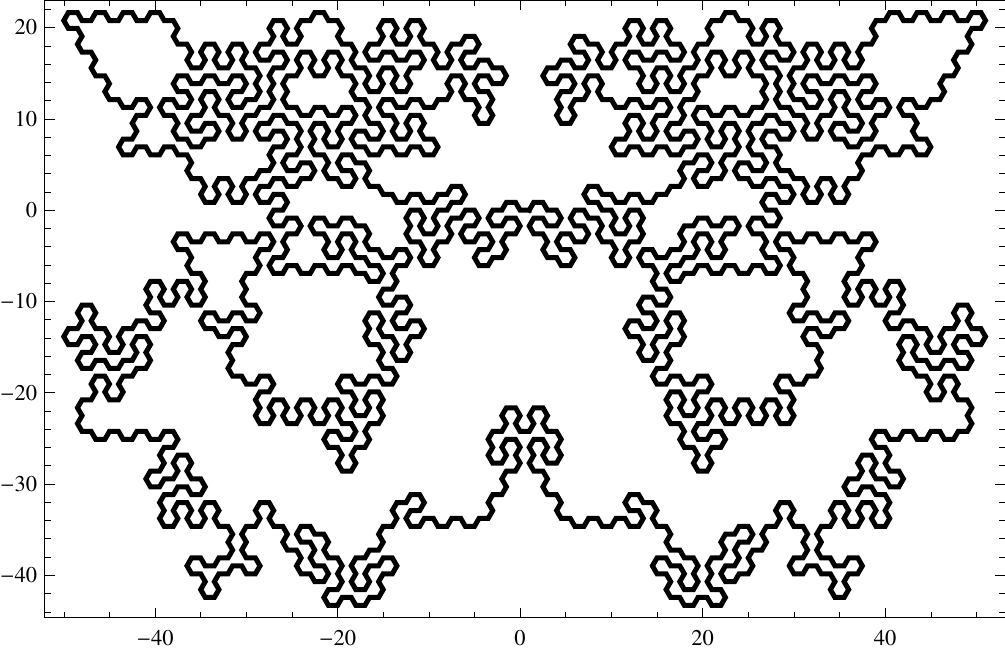}
    \put(41,-5){$\gamma_{4840}$}
    \put(49,42){$\bullet$}
    \end{overpic}
\end{tabular}
\end{center}
\caption{The figure shows the cycles $\gamma_2$, $\gamma_{819}$, and $\gamma_{4840}$ that occur in the $(H,\mathbf{R})$ system. Each cycle is based at the origin, indicated by a black dot, and is symmetric about the line $x=1/2$.}\label{fig:3}
\end{figure}
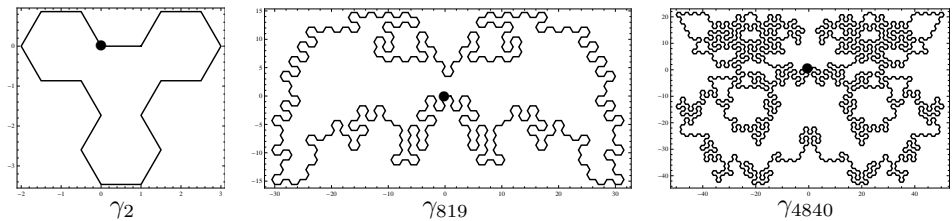

We introduce the $(H,\mathbf{R})$ system at this point in the paper, since it serves as the simplest example of a system with the type of motion we wish to investigate. This motion, which is self-avoiding between returns to the origin, is interesting for a number of reasons. One is that the particle's motion in this system is deterministic. This is in contrast to the large majority of self-avoiding walks, which are generated via some random process \cite{Amit83,Madras88}. Another is that the scattering rule in this system is a local rule. That is, the direction in which particle is scattered at a lattice site, depends only on the orientation of the scatterer at this site. Yet this rule, together with the geometry of the lattice and initial configuration of right rotators, leads to this nonlocal self-avoiding behavior.

To investigate the cycles generated in the $(H,\mathbf{R})$ system, we call the sequence of times $\{\tau_i\}_{i\geq 0}$, given in theorem \ref{thm:1}, the particle's \emph{return times}. Additionally, we refer to the cycle $\gamma_i$ as the particle's $i$-th cycle and let $L(i)=\tau_i-\tau_{i-1}$ denote this cycle's \emph{length}. A number of the particle's cycles that occur in the $(H,\mathbf{R})$ system are shown in figure \ref{fig:3}.

Our first observation is that the cycle lengths $L(i)$, in the $(H,\mathbf{R})$ system, do not exhibit any regular growth. This can be seen in figure \ref{fig:6} (a) where the cycle lengths $L(i)$ are plotted for $0\leq i\leq 154$. Additionally, there are far more short cycles in the particle's trajectory than long cycles. This is shown in figure \ref{fig:6} (b), where the fraction $F(\ell)$ of all cycles of length $\ell$ are shown for the first $t\leq 10^6$ time steps. The approximation $F(\ell)\approx \frac{3}{8}\ell^{-3/2}$ is shown as a dashed line where $\ell\in\{6+4n:n\in\mathbb{N}_0\}$, i.e. each cycle has a length $\ell$, which can be written in the form $6+4n$ with $n$ a nonnegative integer. In particular, $F(6)\approx.38$ so that the cycles of length $\ell=6$ make up nearly forty percent of all cycles up to time $t=10^6$. A list of the first 180 cycle lengths in the $(H,\mathbf{R})$ model is given in Appendix A.

\section{The Time-Averaged Mean Square Displacement}\label{sec:tamsd}

In this section we introduce the notion of a particle's time-averaged mean square displacement in an LLG. We do this for two reasons. First, it allows us to describe the dynamics observed in the $(H,\mathbf{R})$ system, and other LLGs that have been previously studied, in a unified way. Second, it will allow us to compute the average rate at which the particle's displacement in the $(H,\mathbf{R})$ system increases, which will lead us to introduce the idea of \emph{pulsation} in an LLG.

To compare the dynamics of the particle in the $(H,\mathbf{R})$ system to the dynamics observed in other LLGs, we note the following. In previous studies, a typical LLG has a large number of initial configurations. By averaging over this collection of initial configurations it possible to describe the particle's dynamics in terms of the growth of its mean square displacement \cite{Cao97,Grosfils99,Ruijgrok88,Wang94,Wang95.1,Wang95.3,Wang95.2,Kong89,Kong90.2}.

In the $(H,\mathbf{R})$ system there is only a single initial configuration, so the situation is quite different. As we will be show, the particle's mean square displacement cannot be used to describe the dynamics in this system in the same way that it has been used in other LLGs. To formally define the particle's mean square displacement we first need to define the collection of initial configurations over which this average is taken. To do so we introduce the notion of a general LLG model.

Let $(L_{sr},I,C)$ be the \emph{LLG system} on the lattice $L$, with scattering rule $sr$, initial condition $I$ and initial configuration $C$. If $\mathcal{I}$ is a collection of initial conditions and $\mathcal{C}$ a collection of initial configurations, we call $(L_{sr};\mathcal{I},\mathcal{C})$ a \emph{LLG model}. If $I\in\mathcal{I}$ and $C\in\mathcal{C}$, then the system $(L_{sr},I,C)$ is a particular instance of the $(L_{sr};\mathcal{I},\mathcal{C})$ model. The \emph{mean square displacement} of a particle at time $t\geq 0$ in the LLG model $(L_{sr};\mathcal{I},\mathcal{C})$ with position $\mathbf{r}(t)$ is defined as
\begin{equation}\label{eq:msd}
\triangle(t)\equiv\langle|\mathbf{r}(t)-\mathbf{r}(0)|^2\rangle \ \text{for} \ t\geq 0,
\end{equation}
where the average $\langle\cdot\rangle$ is taken over all systems $(L_{sc},I,C)$ for which $I\in\mathcal{I}$ and $C\in\mathcal{C}$. Here, the norm $|\cdot|$ refers to the standard Euclidean distance in the plane.

\begin{figure}
\begin{center}
\begin{tabular}{cc}
\begin{overpic}[scale=.653]{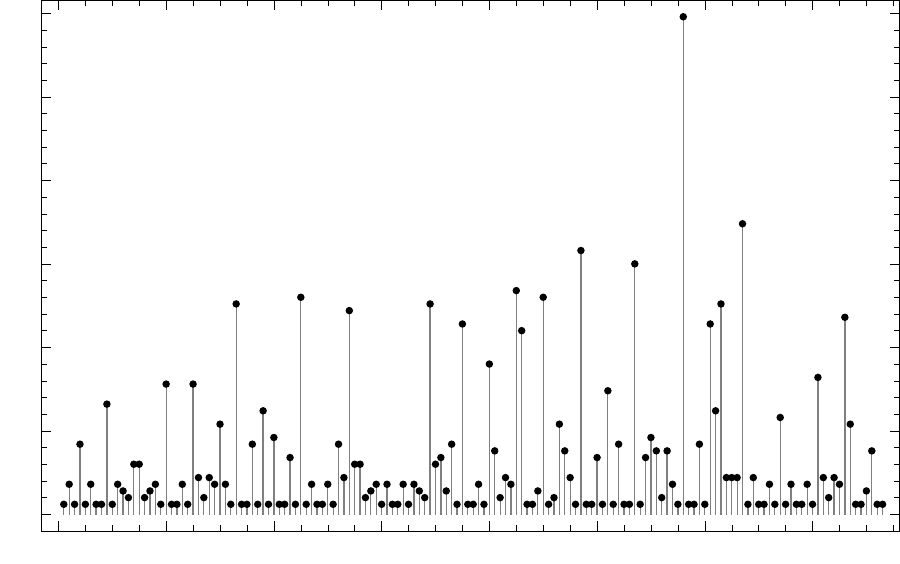}
    \put(5.25,0){\tiny$1$}
    \put(28.5,0){\tiny$40$}
    \put(52,0){\tiny$80$}
    \put(75,0){\tiny$120$}

    \put(1,2.75){\tiny$0$}
    \put(-2.5,23){\tiny$100$}
    \put(-2.5,41.5){\tiny$200$}
    \put(-2.5,60){\tiny$300$}

    \put(27,-6){(a) $L(i)=\tau_i-\tau_{i-1}$}
    \end{overpic}&
    \begin{overpic}[scale=.65]{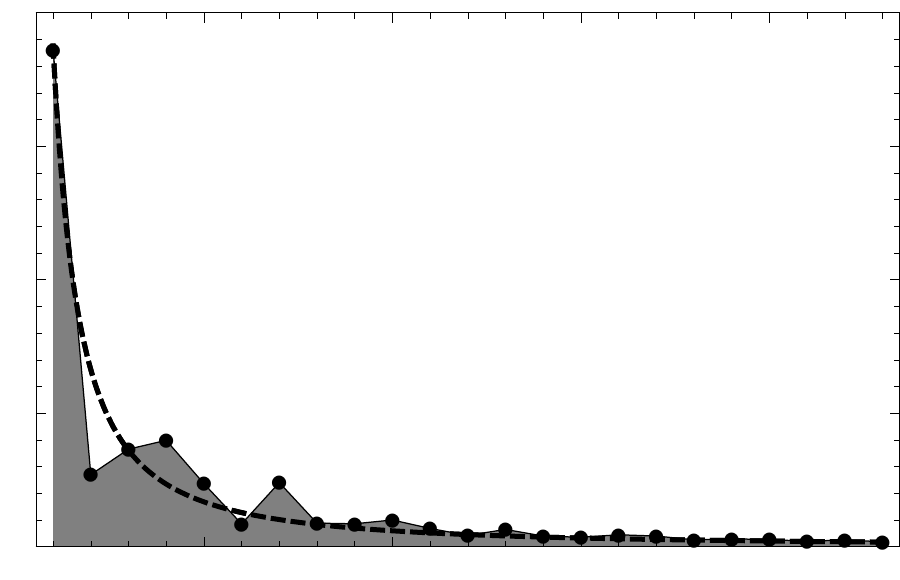}
    \put(5,0){\tiny$6$}
    \put(29,0){\tiny$30$}
    \put(71,0){\tiny$70$}
    \put(96,0){\tiny$94$}

    \put(1,2.75){\tiny$0$}
    \put(0,17.5){\tiny$.1$}
    \put(0,32){\tiny$.2$}
    \put(0,47){\tiny$.3$}
    \put(0,61.5){\tiny$.4$}

    \put(40,-6){(b) $F(\ell)$}
    \end{overpic}
\end{tabular}
\end{center}
\caption{The cycle lengths $L(i)$ for $1\leq i\leq 154$ and the fraction $F(\ell)$ of all cycles of length $\ell$ for the first $t\leq 10^6$ time steps are shown for the $(H,\mathbf{R})$ model in (a) and (b), respectively. The function $F(\ell)\approx \frac{3}{8}\ell^{-3/2}$, shown as the dashed line, indicates a power law decay of the cycle lengths.}\label{fig:6}
\end{figure}
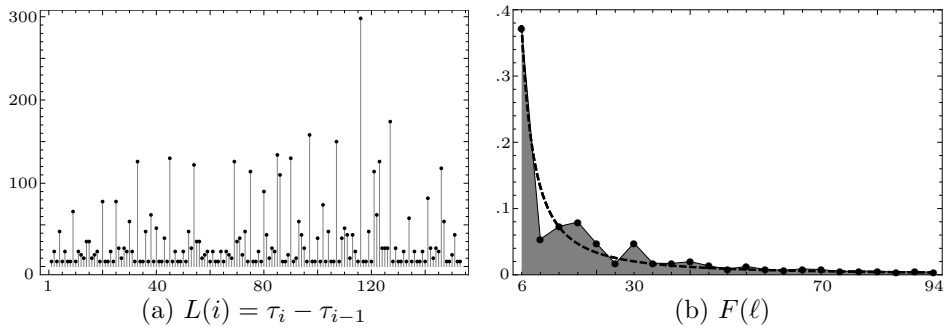

The way the particle's motion in a LLG model is described, in a large number of previous studies, is by the growth of the particle's mean square displacement. To do the same here we introduce the following notation. For the non-negative functions $f(t)$ and $g(t)$, we write
\begin{align}\label{eq:equal}
f(t)\simeq g(t)& \ \text{if} \ \lim_{t\rightarrow\infty}f(t)/g(t)=c \ \text{for some} \ c>0; \ \text{and}\\ \label{eq:unequal}
f(t)\prec g(t)& \ \text{if} \ \lim_{t\rightarrow\infty}f(t)/g(t)=0.
\end{align}
That is, if $f(t)\simeq g(t)$, the growth of $f(t)$ is asymptotically the same as that of $g(t)$. If $f(t)\prec g(t)$, the asymptotic growth of $f(t)$ is dominated by that of $g(t)$.

If the particle's mean square displacement $\triangle(t)\simeq t^{\alpha}$, the particle is said to \emph{diffuse} if $\alpha=1$ and \emph{propagate} if $\alpha=2$, respectively. If $c \prec \triangle(t)\prec t$, for some $c>0$, the particle is said to \emph{subdiffuse} and is said to \emph{superdiffuse} if $t \prec \triangle(t)\prec t^2$. If $\triangle(t)<c$ for all $t\geq 0$ then the particle's trajectory is said to be \emph{bounded}. We note that if the particle subdiffuses, diffuses, superdiffuses, or propagates, the particle has an unbounded trajectory.

Each of these types of motion, i.e. subdiffusion, diffusion, superdiffusion, propagation, as well as bounded motion have been observed in other LLG models \cite{Bunimovich93,Wang94,Wang95.1,Wang95.3,Bunimovich91,Kong90.2,Kong90.1,Meng94,Ziff91}. However, in the $(H;\mathbf{R})$ model, equivalently the $(H,\mathbf{R})$ system, we observe motion, which cannot be described in these terms.

In the $(H;\mathbf{R})$ model, $\triangle(\tau_i)=0$ at each return time $\{\tau_i\}_{i\geq 0}$ to its initial position (see figure \ref{fig:7} (a)). Since there is no last time $\tau_i<\infty$ at which the particle returns to the origin, the particle cannot be said to be either subdiffusing, diffusing, superdiffusing, or propagating. That is, we are not able to describe the particle's dynamics in the $(H;\mathbf{R})$ model in those terms that are typically used to describe other LLG models.

The reason is that each of these terms, e.g. propagation, diffusion, etc. is based on the particle's mean square displacement. A much more appropriate average to consider is the following time-averaged mean square displacement.

\begin{definition}
For the LLG model $(L_{sr};\mathcal{I},\mathcal{C})$, with position $\mathbf{r}(t)$, the quantity
\begin{equation}\label{eq:tamsd}
\bar{\triangle}(t)\equiv\frac{1}{t}\sum_{i=1}^{t}\triangle(i)= \frac{1}{t}\sum_{i=1}^{t}\langle|\mathbf{r}(i)-\mathbf{r}(0)|^2\rangle
\end{equation}
is the particle's \emph{time-averaged mean square displacement} up to time $t$.
\end{definition}

Using the particle's time-averaged mean square displacement $\bar{\triangle}(t)$, we define the following types of dynamics. If $\bar{\triangle}(t)\simeq t^{\alpha}$, we say the particle exhibits \emph{time-averaged diffusion} if $\alpha=1$ and \emph{time-averaged propagation} if $\alpha=2$. If $c \prec \bar{\triangle}(t)\prec t$ for some $c>0$, the particle is said to exhibit \emph{time-averaged subdiffusion} and is said to exhibit \emph{time-averaged superdiffusion} if $t \prec \bar{\triangle}(t)\prec t^2$.

The standard mean square displacement $\triangle(t)$ and the time-averaged mean square displacement $\bar{\triangle}(t)$ of an particle are related in the following way. If a particle exhibits either subdiffusion, diffusion, superdiffusion, or propagation, then the particle also displays time-averaged subdiffusion, diffusion, superdiffusion, or propagation, respectively. However, the converse of this statement does not always hold. This is summarized in the following proposition.

\begin{proposition}\label{thm:0}
For the Lorentz lattice gas model $(L_{sr};\mathcal{I},\mathcal{C})$ the following hold:\\
(a) If $\triangle(t)\simeq t$ then $\bar{\triangle}(t)\simeq t$.\\
(b) If $c\prec\triangle(t)\prec t$ for some $c>0$ then $c\prec\bar{\triangle}(t)\prec t$.\\
(c) If $t\prec\triangle(t)\prec t^2$ then $t\prec\bar{\triangle}(t)\prec t^2$.\\
(d) If $\triangle(t)\simeq t^2$ then $\bar{\triangle}(t)\simeq t^2$.\\
The converse of (a)-(d) do not hold in general.
\end{proposition}

A proof of proposition \ref{thm:0} is given in the Appendix B. We note that since the converse of (a)-(d) in proposition \ref{thm:0} do not always hold, the notions of time-averaged subdiffusion, diffusion, superdiffusion, and propagation are more general than their unaveraged counterparts.

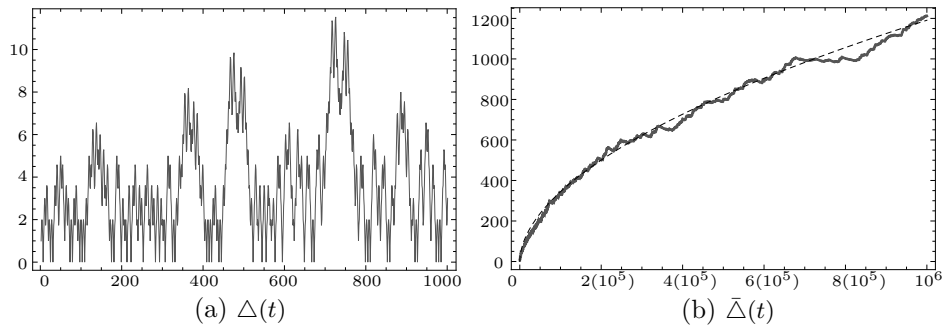
\begin{figure}
\begin{center}
\begin{tabular}{cc}
    \begin{overpic}[scale=.64]{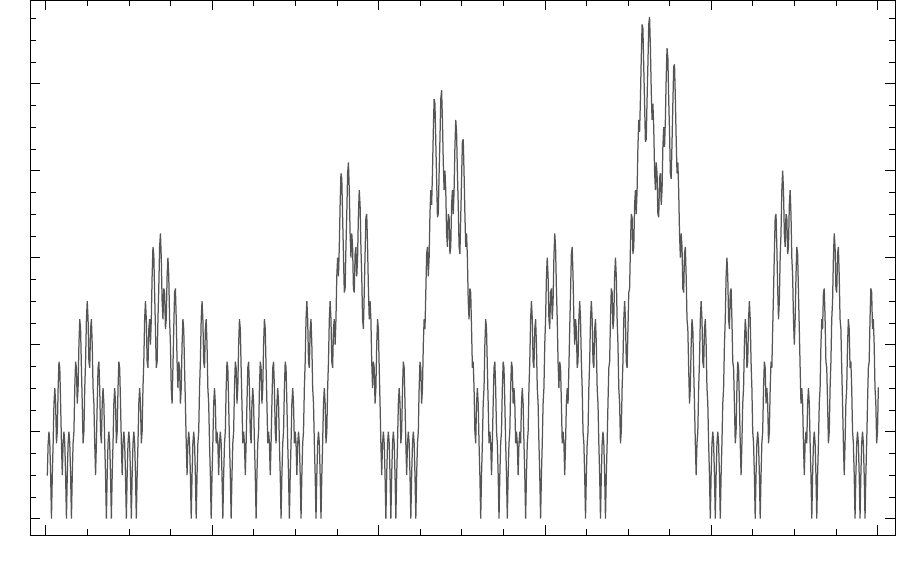}
    \put(40,-7){(a) \small$\triangle(t)$}

    \put(5,0){\tiny$0$}
    \put(21,0){\tiny$200$}
    \put(39,0){\tiny$400$}
    \put(57,0){\tiny$600$}
    \put(76,0){\tiny$800$}
    \put(92,0){\tiny$1000$}

    \put(0,4){\tiny$0$}
    \put(0,14){\tiny$2$}
    \put(0,24){\tiny$4$}
    \put(0,33){\tiny$6$}
    \put(0,43){\tiny$8$}
    \put(-2,52.5){\tiny$10$}
    \end{overpic} &
    \begin{overpic}[scale=.665]{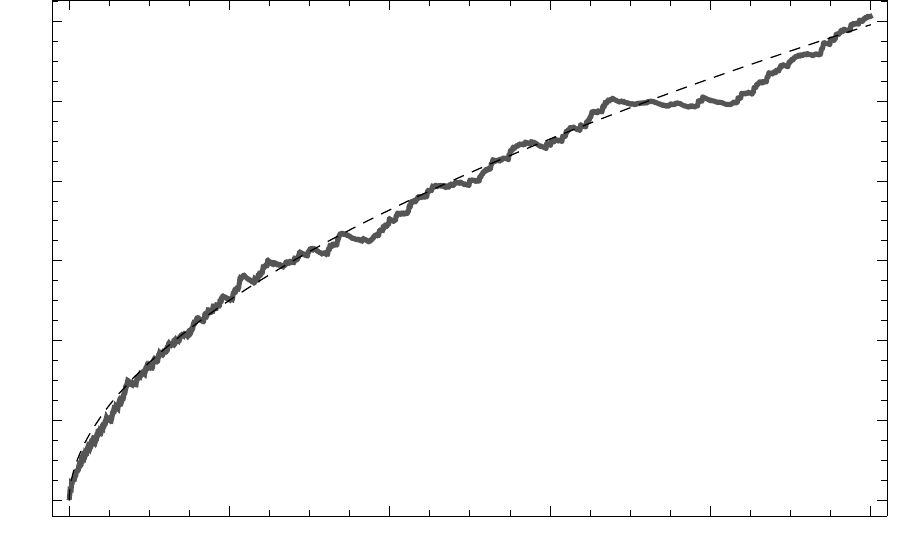}
    \put(43,-7){(b) \small$\bar{\triangle}(t)$}

    \put(5,0){\tiny$0$}
    \put(21,0){\tiny$2(10^5)$}
    \put(39,0){\tiny$4(10^5)$}
    \put(57,0){\tiny$6(10^5)$}
    \put(76,0){\tiny$8(10^5)$}
    \put(94,0){\tiny$10^6$}

    \put(2,4){\tiny$0$}
    \put(-1,13){\tiny$200$}
    \put(-1,21.75){\tiny$400$}
    \put(-1,30.5){\tiny$600$}
    \put(-1,39.5){\tiny$800$}
    \put(-3,48.5){\tiny$1000$}
    \put(-3,57){\tiny$1200$}
    \end{overpic}
\end{tabular}
\end{center}
\caption{The mean square displacement $\triangle(t)$ for $t\leq 10^3$, shown in (a), and the time-averaged mean square displacement $\bar{\triangle}(t)$ for $t\leq10^6$, shown in (b), of the $(H;\mathbf{R})$ model. Here, $\bar{\triangle}(t)\approx \frac{7}{10}t^{7/13}$, which is shown as the dashed line.}\label{fig:7}
\end{figure}

In the $(H;\mathbf{R})$ model, the particle's time-averaged mean square displacement is given by
\begin{equation}
\bar{\triangle}(t)=\frac{1}{t}\sum_{i=1}^t|\mathbf{r}(i)|^2,
\end{equation}
which is the particle's average square distance from the origin from time $1$ to time $t$, cf. equation \eqref{eq:tamsd}. The $(H;\mathbf{R})$ model's mean square displacement $\triangle(t)$ and its time-averaged mean square displacement $\bar{\triangle}(t)$ are compared in figure \ref{fig:7}.

As can be seen in figure \ref{fig:7} (a), the particle's mean square displacement is irregular both in its returns to zero and in its growth away from the horizontal axis.
However, the particle's time-averaged mean square displacement, shown in figure \ref{fig:7} (b), exhibits a fairly regular growth that can be approximated by the power law $\bar{\triangle}(t)\approx \frac{7}{10}t^{7/13}$ for $t\leq 10^6$. One way to interpret this is to say that, the particle displays time-averaged subdiffusion up to this point in time. In fact, as far as is known numerically, the particle always displays time-averaged subdiffusion, which strongly suggests that the particle's trajectory is unbounded.

The time-averaged mean square displacement also allows us to compare the dynamics of the particle in the $(H;\mathbf{R})$ model with previous LLG models, where only the particle's mean square displacement has ever been considered. In particular, any LLG that displays subdiffusive dynamics also displays time-averaged subdiffusion, and in this sense is dynamically similar to the $(H;\mathbf{R})$ model. Additionally, the growth of the particle's time-averaged mean square displacement in the $(H;\mathbf{R})$ model not only suggests that the particle has an unbounded trajectory but also indicates the average rate at which the particle moves away from it initial position.

Under the assumption that the particle's trajectory is unbounded, the particle has what we refer to as a \emph{pulsating} motion.

\begin{definition}\label{def:pulsating}
The particle in an LLG $(L_{sr},I,C)$ is said to \emph{pulsate}, if
there is an infinite sequence of times $\{\tau_i\}_{i\geq 0}$ such that $\mathbf{r}(\tau_i)=\mathbf{r}(0)$ for $i\geq 0$ and the particle's trajectory is unbounded.
\end{definition}

If the particle does pulsate in the $(H;\mathbf{R})$ model, this would have a number of consequences for its sequence of cycles $\{\gamma_i\}_{i\geq 1}$. One is that there would then be an infinite number of distinct cycles. Another is that these cycles would become arbitrarily large as the particle explored more of the lattice. Currently, it is an open question as to whether the particle in the $(H;\mathbf{R})$ model has an unbounded trajectory, although we conjecture this is the case.

Before continuing, we note that a particle with pulsating dynamics has been previously observed in a one-dimensional LLG (see theorem 2 in \cite{Bunimovich04}). The main difference between this and the motion observed in the $(H;\mathbf{R})$ model is that in one-dimension a pulsating particle necessarily intersects its trajectory when returning to its initial position. To the best of our knowledge, the systems considered in this paper are the only known class of LLGs in which the particle has a self-avoiding motion, between returns to its initial position.

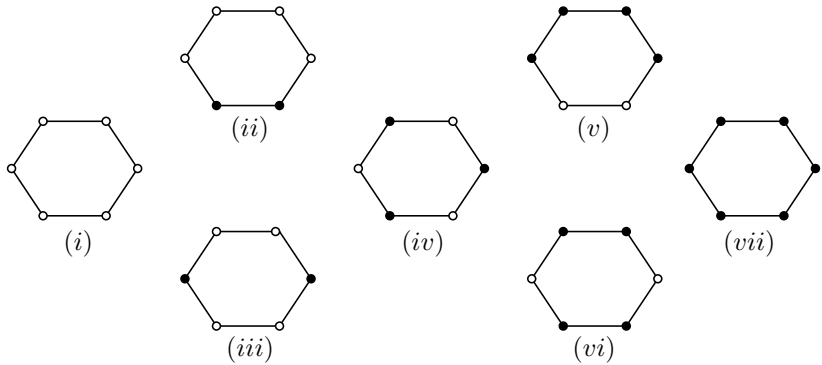
\begin{figure}
    \begin{overpic}[scale=.33]{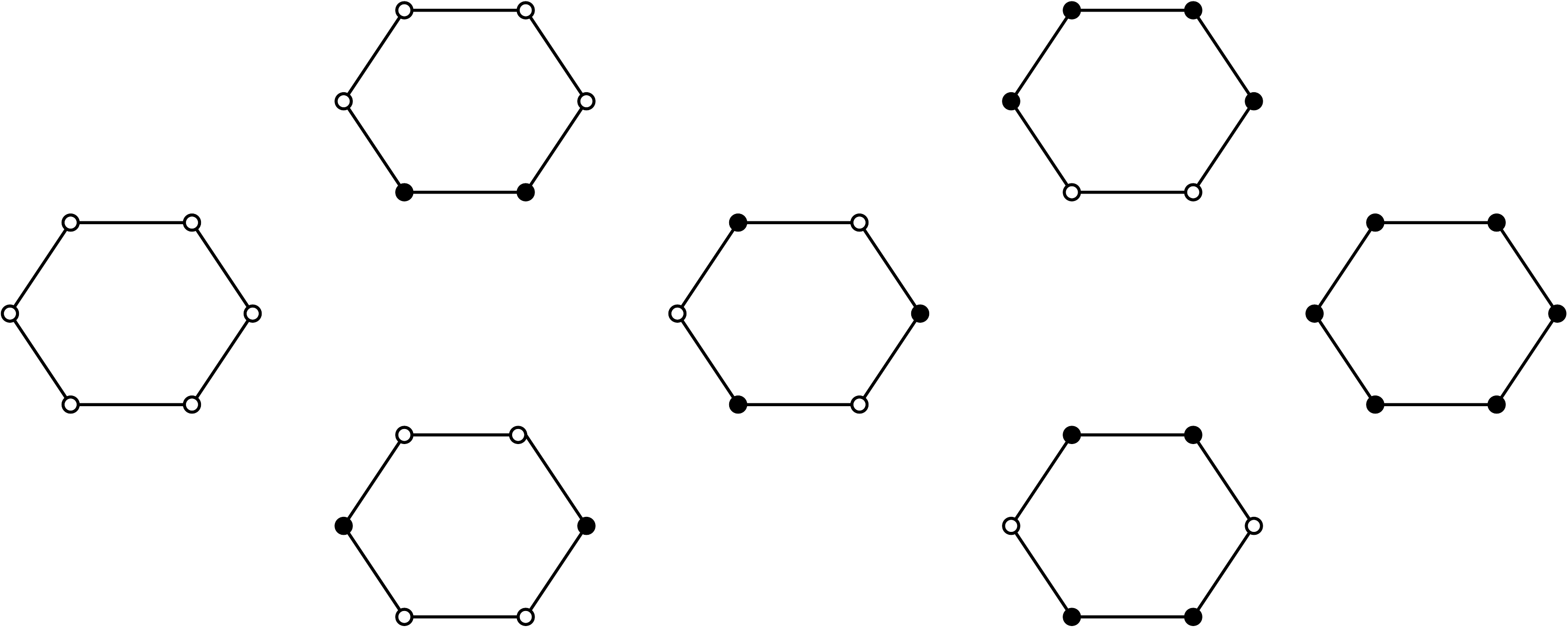}
    \put(7,10){$(i)$}
    \put(27.5,24){$(ii)$}
    \put(27,-3){$(iii)$}
    \put(48.5,10){$(iv)$}
    \put(70,24){$(v)$}
    \put(69.5,-3){$(vi)$}
    \put(88,10){$(vii)$}
    \end{overpic}
\vspace{0.1in}
\caption{The seven \emph{admissible configurations} $\mathcal{A}\subset\mathcal{S}$ on the honeycomb lattice $H$ up to symmetry and reflection.}\label{fig:8}
\end{figure}

\section{Admissible Configurations}\label{sec:5}
In this section, our goal is to discuss the mechanism that causes the self-avoiding behavior found in the $(H;\mathbf{R})$ model. To accomplish this, our strategy is to begin by studying how the particle's movement across a single hexagon of $H$ influences the configuration of scatterers on this hexagon. What we find is that there are two distinct classes of scattering configurations on a hexagon, both of which are invariant to the passage of the particle.

Of these two classes, the first, which we call \emph{admissible}, has the following property. If the particle moves through a hexagon with an admissible configuration, its trajectory on the hexagon will be self-avoiding. Moreover, if two adjacent hexagons have admissible configurations, the particle's trajectory through these two hexagons will also be self-avoiding. By piecing together more and more hexagons we are able to show that a particle's trajectory will remain self-avoiding, away from its initial position, if each hexagon of the lattice has an admissible configuration.

Once this has been shown, we then consider the case in which the particle returns to its initial position. Here, we will prove that after the particle has returned to its initial position, the configuration of any hexagon of the lattice will again be admissible. Hence, the particle will again experience a self-avoiding trajectory until its next return and so on. This is the main result in this section (see theorem \ref{lem:1}). In particular, since the initial configuration of all right scatterers is admissible this will imply that the particle in the $(H;\mathbf{R})$ model has this type of self-avoiding behavior.

Our method for proving this result is to consider the particle's motion as it passes through a finite subset of the honeycomb lattice. To make this notion precise, suppose $\Omega$ is a finite subset of the honeycomb lattice $H$. By way of notation, we write $\mathbf{r}(t)\in\Omega$, if the particle is at a lattice of $\Omega$ at time $t$. If there are two times $t_1<t_2$, such that $\mathbf{r}(t)\in\Omega$ for $t_1\leq t\leq t_2$ and $\mathbf{r}(t_1-1)$, $\mathbf{r}(t_2+1)\notin\Omega$, the particle is said to \emph{enter} $\Omega$ at time $t_1$ and \emph{exit} $\Omega$ at time $t_2$, respectively.

We begin by studying how the particle's movement across a single hexagon $\mathcal{H}\subset H$ changes the hexagon's scattering configuration. Since each lattice site of $\mathcal{H}$ can be either a right or a left scatterer, the number of possible scattering configurations on $\mathcal{H}$ is $2^6=64$. In fact, up to rotation and reflection, there are only thirteen different configurations on $\mathcal{H}$, which are shown collectively in figures \ref{fig:8} and \ref{fig:11}. In these figures, we label each scattering configurations by some element of the set
\begin{equation}
\mathcal{S}=\{(i),(ii),\dots,(xiii)\},
\end{equation}
which is the set of all possible configurations on a single hexagon of $H$.

To understand the particle's influence on these configurations, suppose the particle enters the hexagon $\mathcal{H}$ at $t=t_1$. Then, before it exits the hexagon at $t=t_2$, the particle will cause a change in orientation of a number of scatterers. Therefore, the passage of the particle through $\mathcal{H}$ will cause a change from one configuration of $\mathcal{S}$ to another. If $\alpha\in \mathcal{S}$ is the configuration on $\mathcal{H}$ at time $t=t_1-1$ and $\beta\in \mathcal{S}$ the configuration on $\mathcal{H}$ at time $t=t_2+1$, we say the particle induces a \emph{transition} from $\alpha$ to $\beta$ during this time.

The notion of a transition between two configurations on a hexagon is illustrated in the following example (cf. figure \ref{fig:9}).

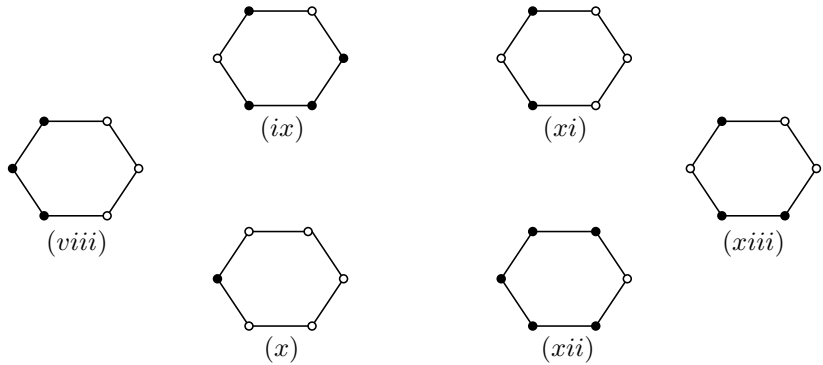
\begin{figure}
    \begin{overpic}[scale=.33]{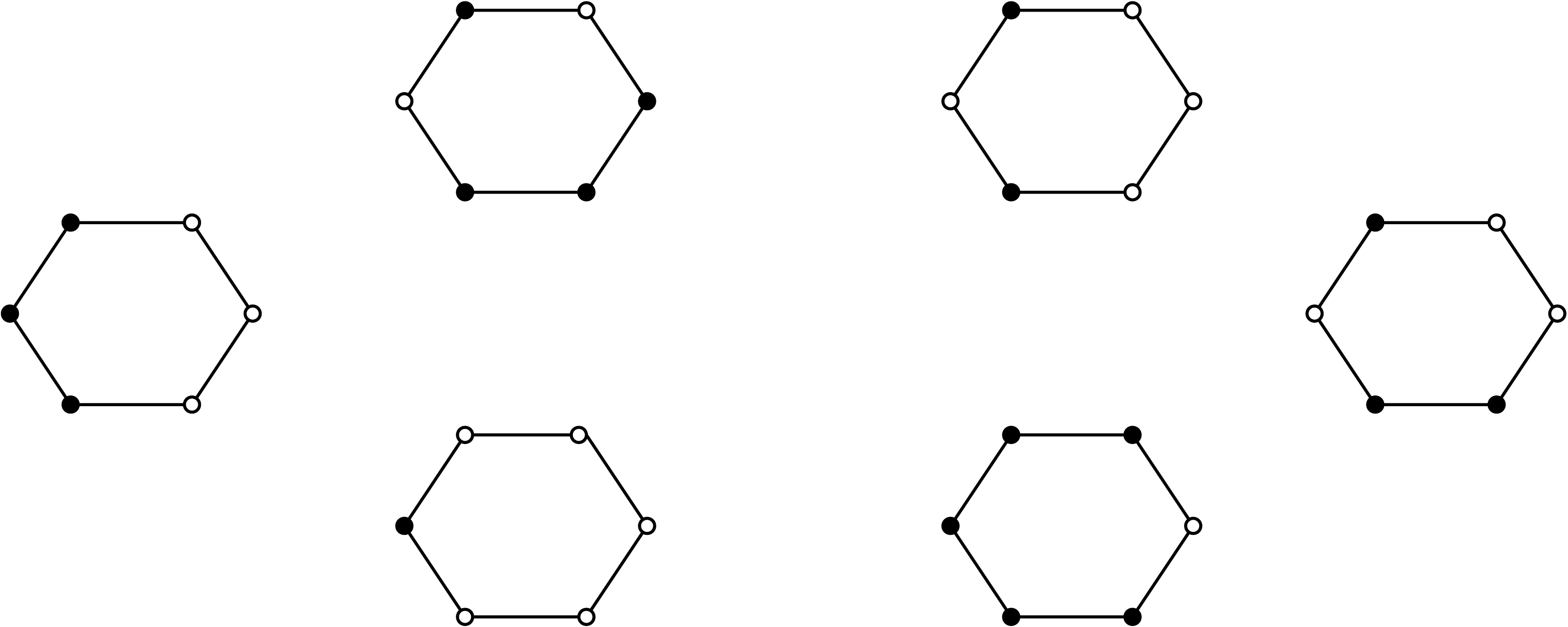}
    \put(4.75,10){$(viii)$}
    \put(31,24){$(ix)$}
    \put(31.5,-3){$(x)$}
    \put(65.75,24){$(xi)$}
    \put(65.5,-3){$(xii)$}
    \put(87.75,10){$(xiii)$}
    \end{overpic}
\vspace{0.1in}
\caption{The six nonadmissible local configurations $\mathcal{N}\subset \mathcal{S}$ on the honeycomb lattice $H$ up to symmetry and reflection.}\label{fig:11}
\end{figure}

\begin{example}\label{ex:1.1}
Consider the single hexagon shown in figure \ref{fig:9} (a), with the configuration $(v)\in \mathcal{S}$. If the particle enters the hexagon from the left at time $t=t_1$, it exits the hexagon at time $t=t_2$ after being scattered at the left, right, and top two lattice sites, as shown in figure \ref{fig:9} (b). In doing so, the particle induces a transition from the configuration $(v)$ to the configuration $(vi)$ on the hexagon.
\end{example}

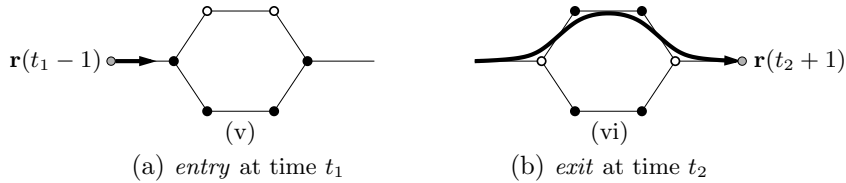
\begin{figure}
    \begin{overpic}[scale=.35]{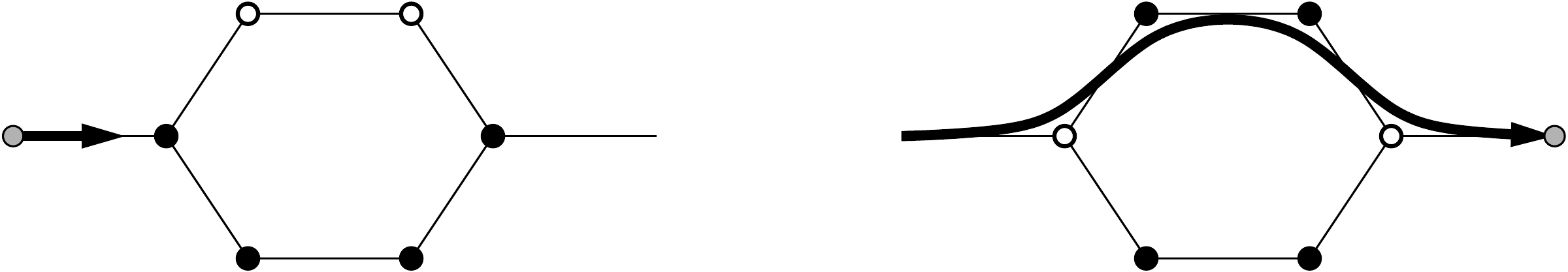}
    \put(4,-9){(a) \small\emph{entry} at time $t_1$}
    \put(18.5,-3.5){\small(v)}
    \put(-15,7.75){\small$\mathbf{r}(t_1-1)$}

    \put(63,-9){(b) \small\emph{exit} at time $t_2$}
    \put(75.5,-3.5){\small(vi)}
    \put(101,7.75){\small$\mathbf{r}(t_2+1)$}
    \end{overpic}
\vspace{0.25in}
\caption{The particle considered in example \ref{ex:1.1} enters a hexagon at time $t_1$ in (a) and exits at time $t_2$ in (b), inducing a transition from the configuration (v) to the configuration (vi).}\label{fig:9}
\end{figure}

If it is possible to transform the configuration $\alpha\in \mathcal{S}$ to the configuration $\beta\in \mathcal{S}$ via a sequence of transitions, we write $\alpha\sim\beta$. This allows us then to state the following theorem, which says that the configurations of $\mathcal{S}$ can be separated into two disjoint sets between which no transitions can occur.

\begin{theorem}\label{thm:partition}
The relation $\sim$ is an equivalence relation that partitions the set of configurations $\mathcal{S}$ into the two subsets
\begin{equation}
\mathcal{A}=\{(i),(ii),\dots,(vii)\} \  \ \text{and} \ \ \mathcal{N}=\{(viii),(ix),\dots,(xiii)\}.
\end{equation}
\end{theorem}

\begin{proof}
Let $\Gamma$ be the graph of all possible transitions between elements of $\mathcal{S}$. That is, $\Gamma$ is the directed graph with vertices labelled by elements of $\mathcal{S}$, where there is a directed edge (arrow) from the configuration $\alpha\in \mathcal{S}$ to the configuration $\beta\in \mathcal{S}$ if and only if $\alpha\sim\beta$.

The graph $\Gamma$ is constructed by checking each possible transition as follows. For a given configuration $\alpha\in \mathcal{S}$ there are six directions from which the particle can enter a hexagon. Hence, there are at most six transitions from any $\alpha\in \mathcal{S}$ to another configuration $\beta\in \mathcal{S}$. Since there are thirteen configurations in $\mathcal{S}$, all transitions between these configurations can be found by exhaustively checking each of these $6\cdot13=78$ possibilities.

The graph $\Gamma$ is shown in figure \ref{fig:12}, from which one can immediately check that the relation $\sim$ is reflexive and symmetric. Transitivity follows from the fact that if $\alpha\sim\beta$ and $\beta\sim\delta$ then there is a sequence of transitions that transform the configuration $\alpha$ into $\delta$. Since $\Gamma$ has two connected components containing the elements of $\mathcal{A}$ and $\mathcal{N}$ respectively, this completes the proof.
\end{proof}

An important consequence of theorem \ref{thm:partition} is that the particle can only induce a transition between elements of $\mathcal{A}$ or elements of $\mathcal{N}$, but not between these sets. For example, if a hexagon begins with the configuration $\alpha\in\mathcal{A}$ then the particle's motion through the hexagon can only change the hexagon's configuration to another element of $\mathcal{A}$. Phrased another way, the configurations $\mathcal{A}$ and $\mathcal{N}$ are invariant with respect to the passage of the particle through the hexagon.

In this paper, our focus is on the set of configurations $\mathcal{A}$, which we use to define the following concept.

\begin{definition}\label{def:admiss}
We call the set of configurations $\mathcal{A}\subset \mathcal{S}$ the \emph{admissible configurations} of a hexagon. We say the configuration $C$, on the honeycomb lattice $H$, is \emph{admissible} if this configuration restricted to each hexagon of $H$ is an admissible configuration.
\end{definition}

An example of an admissible configuration is the configuration of all right scatterers, found in the $(H;\mathbf{R})$ model. The reason this configuration is admissible is that its restriction to any hexagon of $H$ is the admissible configuration $(i)\in\mathcal{A}$, shown in figure \ref{fig:8}.

Having defined the notion of an admissible configuration, we now consider the effect this type of configuration has on the particle's motion. We begin by studying the particle's trajectory restricted to a subset of $H$. To be precise about what we mean, suppose $\Omega$ is a finite subset of $H$ with a given configuration of scatterers. If the particle enters $\Omega$ at time $t_1$ and exits at time $t_2$ we call the sequence of positions $\{\mathbf{r}(t):t_1 \leq t\leq t_2\}$ a \emph{local crossing} of $\Omega$ (cf. figure \ref{fig:9}). Additionally, we say the sequence $\gamma=\{\mathbf{r}(t):t_1\leq t\leq t_2\}$ is a \emph{local cycle} of $\Omega$ if $\gamma\subset\Omega$ is a cycle and has the additional property that $\mathbf{v}(t_1)=\mathbf{v}(t_2)$, i.e. $\mathbf{r}(t_1+1)=\mathbf{r}(t_2+1)$.

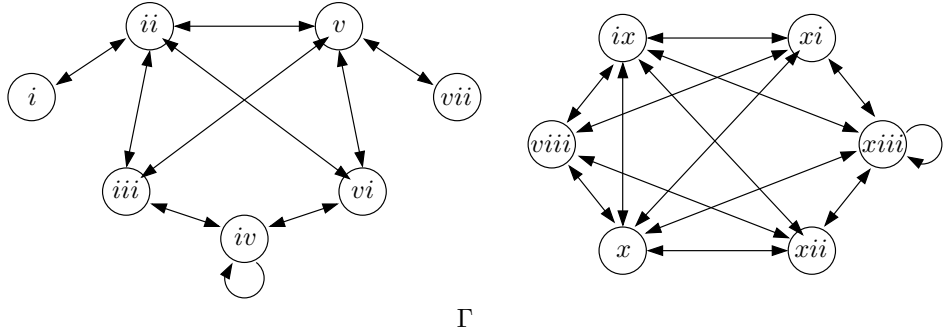
\begin{figure}
    \begin{overpic}[scale=.99]{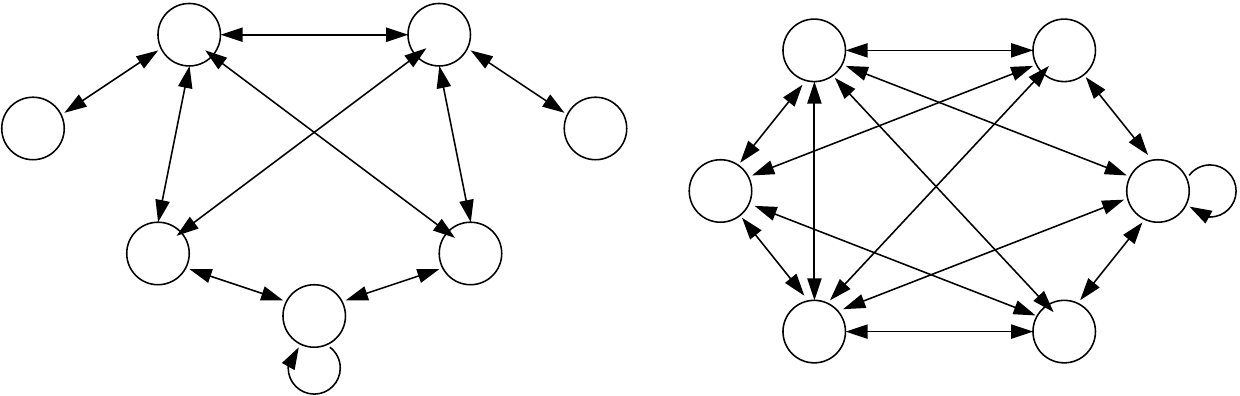}
    \put(2.25,20.75){$i$}
    \put(14.25,28.5){$ii$}
    \put(11.25,10.75){$iii$}
    \put(24.25,5.75){$iv$}
    \put(34.75,28.5){$v$}
    \put(36.75,10.75){$vi$}
    \put(46.25,20.75){$vii$}

    \put(55.85,15.75){$viii$}
    \put(64.5,27.25){$ix$}
    \put(65,4.5){$x$}
    \put(84.5,27.25){$xi$}
    \put(84,4.5){$xii$}
    \put(91.2,15.75){$xiii$}


    \put(48,-3){$\Gamma$}
    \end{overpic}
\vspace{0.1in}
\caption{The \emph{graph of transitions} $\Gamma$ of all transitions that are possible between all elements of the configurations $\mathcal{S}=\mathcal{A}\cup\mathcal{N}$. Shown to the left are all the possible transitions between elements of $\mathcal{A}=\{(i),\dots,(vii)\}$ and right, all the transitions between elements of $\mathcal{N}=\{(viii),\dots,(xiii)\}$, respectively.}\label{fig:12}
\end{figure}

To understand the notion of a local cycle, suppose $\gamma=\{\mathbf{r}(t):0\leq t\leq t_2\}$ is a local cycle of the system $(H,I,C)$ where $I=(\mathbf{r},\mathbf{v})$. Moreover, consider the particle's trajectory in the related system $(H,J,C)$ where $J=(\mathbf{r}(\tau),\mathbf{v}(\tau))$ for some $0<\tau<t_2$. That is, the particle in the $(H,J,C)$ system starts at the point $\mathbf{r}(\tau)$ on $\gamma$ and moves along this cycle. When it reaches the base $\mathbf{r}=\mathbf{r}(0)$ of $\gamma$, implies that $\gamma$ is a local cycle, i.e. $\mathbf{v}(0)=\mathbf{v}(t_2)$, means that the particle continues along $\gamma$ until it reaches its initial position $\mathbf{r}(\tau)$.

Thus, if $\mathbf{s}(t)$ is the particle's position in the $(H,J,C)$ system, $\delta=\{\mathbf{s}(t):0\leq t\leq t_2\}$ is a cycle. Moreover, $\delta=\gamma$ as a set of positions. However, $\gamma$ and $\delta$ are not the same cycle, since they have different bases. In this sense, the local cycle $\gamma$ is what we call \emph{baseless}, since beginning at any other position on the cycle $\gamma$, the particle will move along the same set of lattice sites.

If $\gamma$ is not a local cycle this will not be the case. Therefore, a local and a nonlocal cycle are different, as is illustrated in the following example.

\begin{example}\label{ex:localcyc}
Consider the particle that moves along the hexagon in figure \ref{fig:localcyc} (a), starting at $\mathbf{h}_1$. Since each scatterer on the hexagon is a right scatterer, the sequence of positions $\gamma=\{\mathbf{r}(t):0 \leq t \leq 6\}$ is a cycle. Moreover, since $\mathbf{v}(0)=\mathbf{v}(6)$ then $\gamma$ is a local cycle. To see that $\gamma$ is baseless, suppose instead that the particle starts at $\mathbf{h}_2$, as is shown in \ref{fig:localcyc} (b). In this case, the particle moves along the cycle $\delta$ up to time $t=6$. As can be seen, $\gamma=\delta$ as a set of positions, although these cycles have different bases.

Now suppose the particle moves along the hexagon in figure \ref{fig:localcyc} (c), starting at $\mathbf{h}_1$. In this case, each scatterer is a right scatterer except the scatterer at $\mathbf{h}_1$. Therefore, $\mu=\{\mathbf{r}(t):0 \leq t \leq 6\}$ is a cycle but is not a local scatterer as $\mathbf{v}(0)\neq\mathbf{v}(6)$. If the particle instead starts at $\mathbf{h}_2$, as is shown in \ref{fig:localcyc} (d), the particle exits the hexagon after three time steps. That is, if $\nu$ is the particle's trajectory up to this point, $\mu\neq\nu$. The reason these are different is that $\mu$ is not a local cycle.
\end{example}

With the notion of a local crossing and a local cycle in place, we now combine the two. If $C$ is a configuration of scatterers on $H$ and $\Omega$ a subset of $H$, we let $\mathcal{L}_{\Omega}(C)$ denote the set of all local crossings and local cycles on $\Omega$. We call $\mathcal{L}_{\Omega}(C)$ the set of \emph{local trajectories} on $\Omega$ with the configuration $C$. In what follows, we typically let the trajectories $\mathcal{L}_{\Omega}(C)=\{\ell_1,\dots,\ell_N\}$ so that each local crossing or local cycle of $\Omega$ is denoted by some $\ell_i$.

\begin{figure}
    \begin{overpic}[scale=.425]{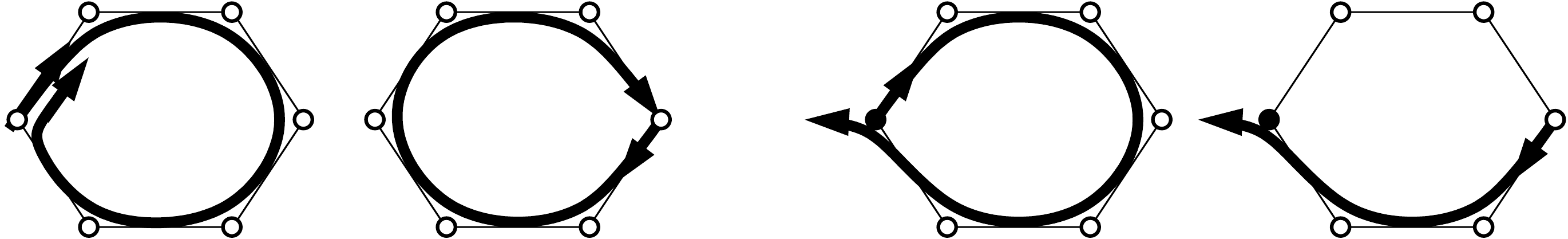}
    \put(-3,11){\small$\mathbf{v}(0)$}
    \put(5,7){\small$\mathbf{v}(6)$}
    \put(-1.5,4){$\mathbf{h}_1$}

    \put(41.5,3){\small$\mathbf{v}(0)$}
    \put(42,9){$\mathbf{h}_2$}

    \put(7,-4){(a) $\gamma$}
    \put(29.5,-4){(b) $\delta$}
    \put(62,-4){(c) $\mu$}
    \put(87,-4){(d) $\nu$}

    \put(50,9.5){\small$\mathbf{v}(6)$}
    \put(59,8){\small$\mathbf{v}(0)$}
    \put(53.5,4){$\mathbf{h}_1$}

    \put(99,9){$\mathbf{h}_2$}
    \put(98.5,3){\small$\mathbf{v}(0)$}
    \end{overpic}
\vspace{0.15in}
\caption{Shown in (a) is the local cycle $\gamma$ with base $\mathbf{h}_1$. In (b) is the cycle $\delta$ is shown, which is the cycle $\gamma$ but with the base $\mathbf{h}_2$. The cycle $\mu$, shown in (c), is a nonlocal cycle since it is not the same as the trajectory $\nu$ shown in (d).}\label{fig:localcyc}
\end{figure}
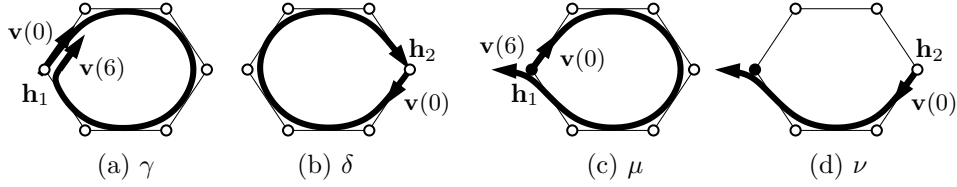

To show that admissible configurations lead to self-avoiding motion, we will heavily rely on this notion of the local trajectories $\mathcal{L}_{\Omega}(C)$. In particular, we will need to partition these trajectories into three sets in a specific way that will allow us to inductively piece together a trajectory that is self-avoiding.

To do this, we first need to define the type of partition we will consider. A \emph{partition} of the local trajectories $\mathcal{L}_{\Omega}(C)$ into the three sets $P=\{P_1,P_2,P_3\}$ is a function $\chi:\mathcal{L}_{\Omega}(C)\rightarrow P$. For the partition $\chi$ we let $\mathcal{L}^k_{\Omega}(C)=\{\ell\in\mathcal{L}_{\Omega}(C):\chi(\ell)=P_k\}$ for $k=1,2,3$. This will allow us to define the notion of a triperfect partition, which will be our main tool for showing why admissible configurations lead to self-avoiding behavior.

\begin{definition}\label{def:tri}
Let $\chi:\mathcal{L}_{\Omega}(C)\rightarrow P$ be a partition of the local trajectories $\mathcal{L}_{\Omega}(C)$. If the partition has the following properties \ref{property:1}-\ref{property:4}, we call it a \emph{triperfect partition}.

\begin{property}\label{property:1}\textbf{Local Self-Avoiding Property}: Let $\ell=\{\mathbf{r}(t_1),\dots,\mathbf{r}(t_2)\}\in\mathcal{L}_{\Omega}(C)$. Then
$\mathbf{r}(t)\neq\mathbf{r}(\tau)$ for all $t_1\leq t<\tau\leq t_2$ if $\ell$ is a local crossing of $\Omega$.
\end{property}

\begin{property}\label{property:4}\textbf{Site Partition Property}: For each lattice site $\mathbf{h}\in\Omega$ there are exactly three local trajectories $\ell_1,\ell_2,\ell_3\in\mathcal{L}_{\Omega}(C)$ containing $\mathbf{h}$ where $\ell_k\in\mathcal{L}^k_{\Omega}(C)$ for each $k=1,2,3$.
\end{property}
\end{definition}

If the sets $\mathcal{L}^k_{\Omega}(C)$ for $k=1,2,3$ form a triperfect partition of $\mathcal{L}_{\Omega}(C)$ then we call these sets the \emph{local families} on $\Omega$ with the configuration $C$. If we do have a triperfect partition of  $\mathcal{L}_{\Omega}(C)$, property \ref{property:1} states that any local crossing on $\Omega$ is self-avoiding, and property \ref{property:4} that exactly one trajectory from each local family goes through each lattice site of $\Omega$.

Having defined the notion of a triperfect partition for a general subset $\Omega
\subset H$, we now consider the specific case in which $\Omega$ is a single hexagon. What we find it that it is always possible to create a triperfect partition of a hexagon's local trajectories, if the hexagon's configuration is admissible. This is stated in the following lemma.

\begin{lemma}\label{lem:tri}
There exists a triperfect partition of the local trajectories $\mathcal{L}_{\mathcal{H}}(C)$ for any admissible configuration $C$ and hexagon $\mathcal{H}\subset H$.
\end{lemma}

\begin{proof}
Suppose $C$ is an admissible configuration on $H$. By definition, the configuration $C|_{\mathcal{H}}$, which is the restriction of $C$ to the hexagon $\mathcal{H}$, is an element of $\mathcal{A}$. If $C|_{\mathcal{H}}=\alpha$, we will then write the set of local trajectories $\mathcal{L}_{\mathcal{H}}(C)$ as $\mathcal{L}_{\mathcal{H}}(\alpha)$.

To demonstrate that $\mathcal{L}_{\mathcal{H}}(\alpha)$ has a triperfect partition for each $\alpha\in\mathcal{A}$, consider the local trajectories shown in figure \ref{fig:13}. In this figure each local trajectory is given a specific line type to indicate which local family it belongs to: either \emph{dashed-dotted black}, \emph{dashed gray}, or \emph{solid light gray}. One can check that this scheme partitions each set of local trajectories $\mathcal{L}_{\mathcal{H}}(\alpha)$ into three distinct families $\mathcal{L}^k_{\mathcal{H}}(\alpha)$ for $k=1,2,3$ for which properties \ref{property:1}-\ref{property:4} hold.
\end{proof}

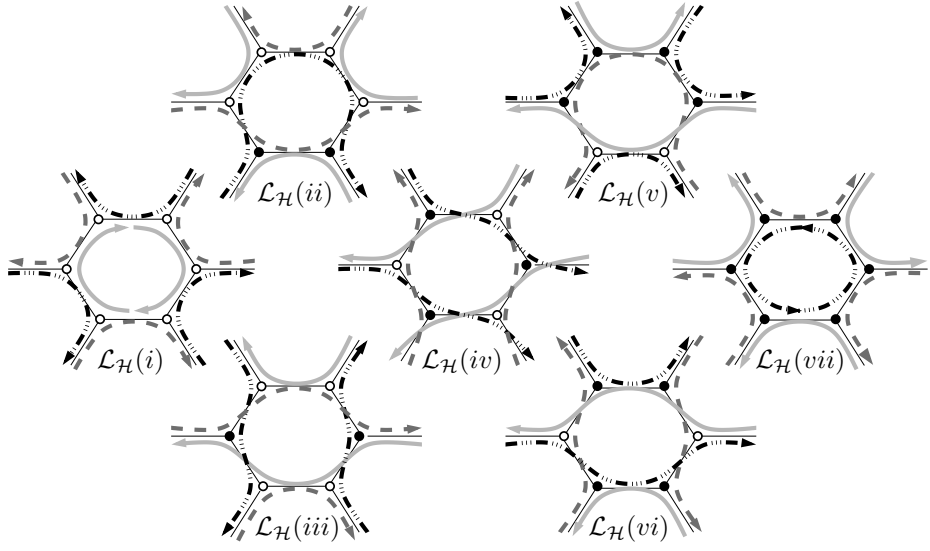
\begin{figure}
    \begin{overpic}[scale=.35]{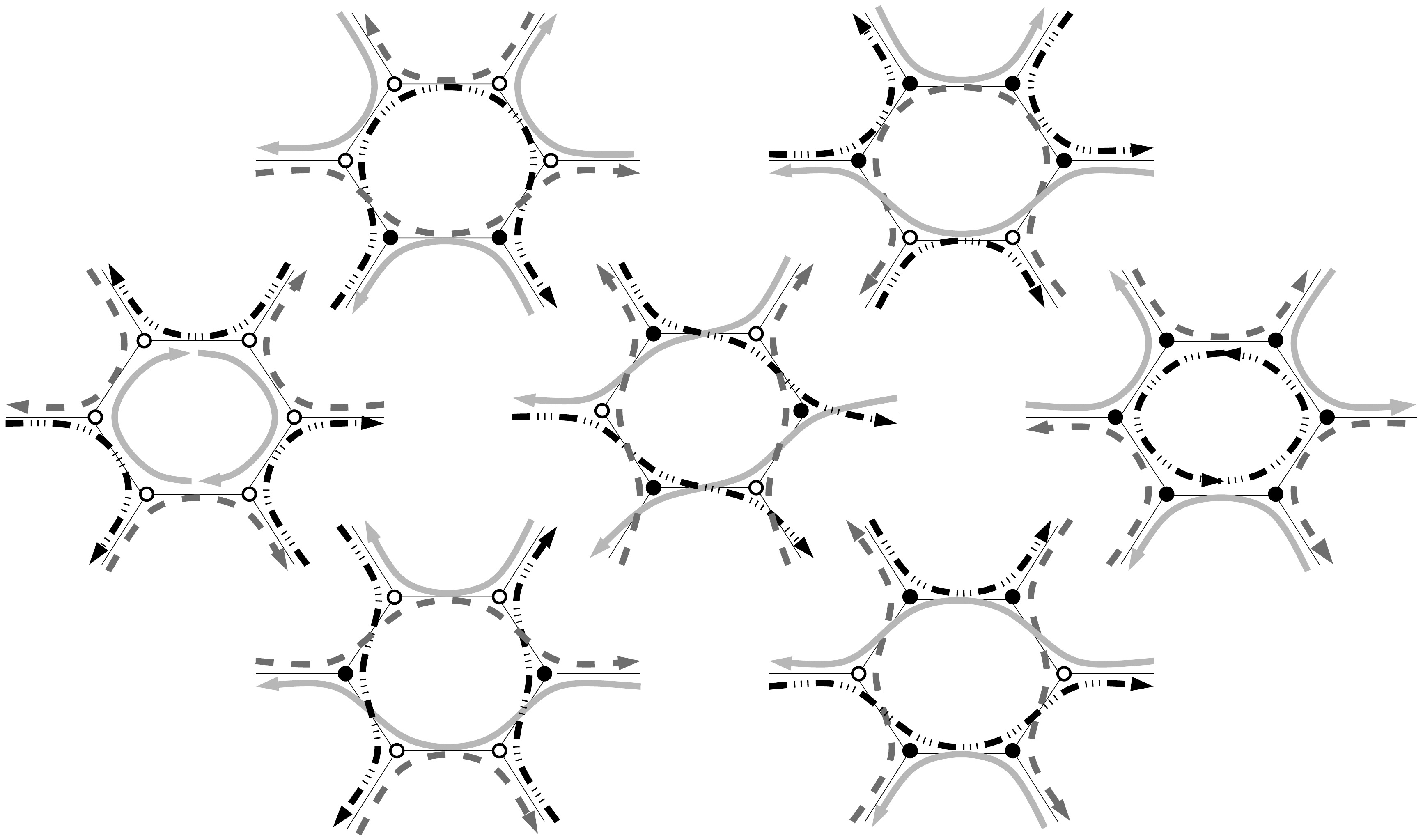}
    \put(10,19){$\mathcal{L}_{\mathcal{H}}(i)$}
    \put(27.25,37){$\mathcal{L}_{\mathcal{H}}(ii)$}
    \put(27,1){$\mathcal{L}_{\mathcal{H}}(iii)$}
    \put(45,19){$\mathcal{L}_{\mathcal{H}}(iv)$}
    \put(64,37){$\mathcal{L}_{\mathcal{H}}(v)$}
    \put(63,1){$\mathcal{L}_{\mathcal{H}}(vi)$}
    \put(81,19){$\mathcal{L}_{\mathcal{H}}(vii)$}
    \end{overpic}
\caption{A triperfect partition of the local trajectories $\mathcal{L}_{\mathcal{H}}(\alpha)$ on a hexagon $\mathcal{H}$ is shown for each configuration $\alpha\in\mathcal{A}$. Trajectories in the same family $\mathcal{L}^k_{\mathcal{H}}(\alpha)$ for $k=1,2,3$ are given the same line type: either \emph{dashed-dotted black}, \emph{dashed gray}, or \emph{solid light gray}, respectively.}\label{fig:13}
\end{figure}

Our strategy is to use lemma \ref{lem:tri} to create a triperfect partition of $\mathcal{L}_{\Omega}(C)$, where $\Omega$ consists of a number of hexagons. In fact, our goal is to show that the trajectories $\mathcal{L}_{H}(C)$ over the entire lattice $H$ have a triperfect partition if $C$ is admissible. To do this, we will need a way of creating a triperfect partition from two smaller partitions.

With this in mind, suppose $\chi:\mathcal{L}_{\Omega}(C)\rightarrow P$ partitions the local trajectories $\mathcal{L}_{\Omega}(C)$. Then, for any subset $\Psi\subseteq\Omega$, the map $\chi$ \emph{induces} a partition $\overline{\chi}:\mathcal{L}_{\Psi}(C)\rightarrow P$ given by $\overline{\chi}(\ell_*)=\chi(\ell)$ where $\ell_*=\ell|_{\Psi}$. We call the partition $\bar{\chi}$ the \emph{induced partition} of $\chi$ on the set of local trajectories $\mathcal{L}_\Psi(C)$. For $i=1,2$ suppose there are sets $\Omega_i\subset H$ and partitions $\chi_i:\mathcal{L}_{\Omega_i}(C)\rightarrow P$. Then the partitions $\chi_1$ and $\chi_2$ are said to \emph{agree} on the set $\Psi\subseteq\Omega_1\cup\Omega_2$ if $\overline{\chi}_1(\ell)=\overline{\chi}_2(\ell) \ \text{for all} \  \ell\in\mathcal{L}_{\Psi}(C)$, where $\overline{\chi}_i$ is the induced partition of $\chi_i$ on the set $\Psi$ for $i=1,2$.

One of the main ideas, in this section is that, if two different subsets of $H$ have a triperfect partition, these partitions can be merged into a single triperfect partition. This method, which we will use in the proof of lemma \ref{thm:noncross}, is demonstrated in the following example.

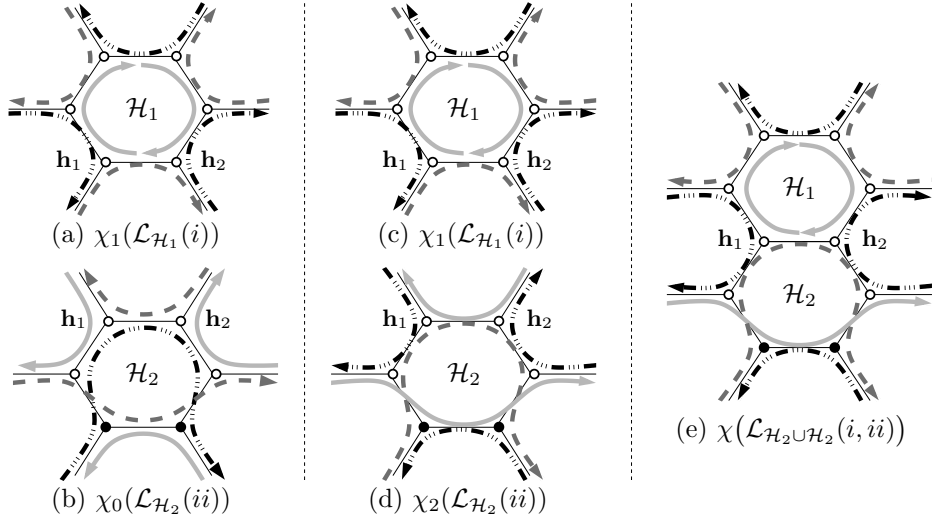
\begin{figure}
    \begin{overpic}[scale=.37]{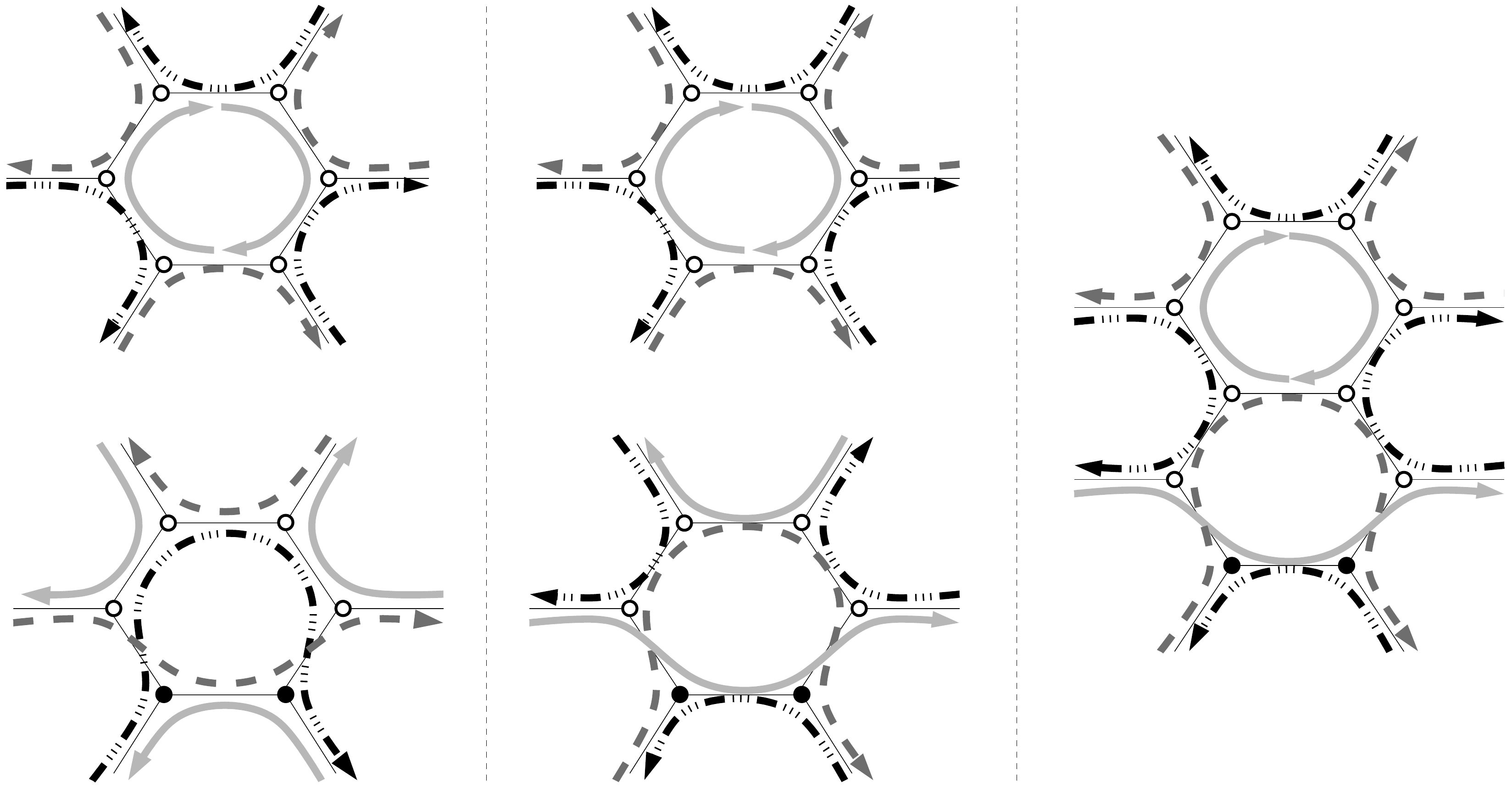}
    \put(5,26){(a) $\chi_1(\mathcal{L}_{\mathcal{H}_1}(i))$}
    \put(5.5,34){\small$\mathbf{h}_1$}
    \put(21,34){\small$\mathbf{h}_2$}

    \put(5,-2.5){(b) $\chi_0(\mathcal{L}_{\mathcal{H}_2}(ii))$}
    \put(6,17){\small$\mathbf{h}_1$}
    \put(21.5,17){\small$\mathbf{h}_2$}

    \put(40,26){(c) $\chi_1(\mathcal{L}_{\mathcal{H}_1}(i))$}
    \put(40.5,34){\small$\mathbf{h}_1$}
    \put(56,34){\small$\mathbf{h}_2$}

    \put(39,-2.5){(d) $\chi_2(\mathcal{L}_{\mathcal{H}_2}(ii))$}
    \put(40.5,17){\small$\mathbf{h}_1$}
    \put(56,17){\small$\mathbf{h}_2$}

    \put(72,5){(e) $\chi\big(\mathcal{L}_{\mathcal{H}_2\cup\mathcal{H}_2}(i,ii)\big)$}
    \put(76.5,25.5){\small$\mathbf{h}_1$}
    \put(92,25.5){\small$\mathbf{h}_2$}

    \put(13,39.5){$\mathcal{H}_1$}
    \put(13,11){$\mathcal{H}_2$}

    \put(47.5,39.5){$\mathcal{H}_1$}
    \put(47.5,11){$\mathcal{H}_2$}

    \put(83.5,31){$\mathcal{H}_1$}
    \put(83.5,19.5){$\mathcal{H}_2$}
    \end{overpic}
    \vspace{0.1in}
\caption{Initially, the local trajectories of $\mathcal{L}_{\mathcal{H}_1}(i)$ and $\mathcal{L}_{\mathcal{H}_2}(ii)$ are partitioned by $\chi_1$ and $\chi_0$, shown respectively in (a) and (b). The partition $\chi_1$ is then used to generate a partition $\chi_2$ on $\mathcal{L}_{\mathcal{H}_2}(ii)$ in (d). The two partitions are then merged into the single partition $\chi$ on $\mathcal{L}_{\mathcal{H}_1\cup\mathcal{H}_2}(i,ii)$ in (e) as described in example \ref{ex:0}.}\label{fig:14}
\end{figure}

\begin{example}\label{ex:0}
In figure \ref{fig:14} (a) we consider the hexagon $\mathcal{H}_1$ with the admissible configuration $(i)$ and triperfect partition $\chi_1:\mathcal{L}_{\mathcal{H}_1}(i)\rightarrow P$. In figure \ref{fig:14} (b) we furthermore consider the hexagon $\mathcal{H}_2$ with the admissible configuration $(ii)$ and triperfect partition $\chi_0:\mathcal{L}_{\mathcal{H}_2}(ii)\rightarrow P$. We note that these two partitions are the same as those given in figure \ref{fig:13} for the hexagon with configuration $(i)$ and $(ii)$, respectively. As in figure \ref{fig:13}, local trajectories with the same line type: either \emph{dashed-dotted black}, \emph{dashed gray}, or \emph{solid gray}, belong to the same local family.

As is shown in figure \ref{fig:14}, the two hexagons share the lattice sites $\mathbf{h}_1,\mathbf{h_2}\in\mathbb{H}$. However, on the set $\Psi=\mathcal{H}_1\cap\mathcal{H}_2$ the partitions $\chi_1$ and $\chi_0$ do not agree. The question is whether there is a triperfect partition of $\mathcal{L}_{\mathcal{H}_2}(ii)$ that does agree with $\chi_1$ on $\Psi$. As we will show, this is indeed the case. Our strategy is to use the partition $\chi_1$ to generate a new triperfect partition of $\mathcal{L}_{\mathcal{H}_2}(ii)$. To do so, we let $\ell_{ij}(\mathcal{H}_m)$ be the unique local trajectory on $\mathcal{H}_m$ that first passes through $\mathbf{h}_i$ then $\mathbf{h}_j$, for $i,j,m=1,2$ where $i\neq j$.

First, note that the local family on $\mathcal{H}_1$ that contains $\ell_{12}(\mathcal{H}_1)$ in figure \ref{fig:14} (a) has the line type \emph{dashed gray}, whereas the local family on $\mathcal{H}_2$ that contains $\ell_{12}(\mathcal{H}_2)$ in figure \ref{fig:14} (b) has the line type \emph{dashed-dotted black}. So that $\ell_{12}(\mathcal{H}_1)$ and $\ell_{12}(\mathcal{H}_2)$ have the same line type we change the local family on $\mathcal{H}_2$ in \ref{fig:14} (b) that contains $\ell_{12}(\mathcal{H}_2)$ so that it has the line type \emph{dashed gray}. Next, so that $\ell_{21}(\mathcal{H}_1)$ and $\ell_{21}(\mathcal{H}_2)$ have the same line type, we likewise change the local family on $\mathcal{H}_2$ that contains $\ell_{21}(\mathcal{H}_2)$ so that it has the line type \emph{solid gray}. We give the remaining local family on $\mathcal{H}_2$ the line type \emph{dashed gray}. The resulting partition $\chi_2:\mathcal{L}_{\mathcal{H}_2}(ii)\rightarrow P$ is shown in figure \ref{fig:14} (d).

As one can check the partitions shown in figures \ref{fig:14} (c) and (d) agree on $\Psi$ and can be merged into the single triperfect partition $\chi:\mathcal{L}_{\mathcal{H}_1\cup\mathcal{H}_2}(i,ii)\rightarrow P$ shown in figure \ref{fig:14} (e). Here, the notation $\mathcal{L}_{\mathcal{H}_1\cup\mathcal{H}_2}(i,ii)$ denotes the set of local trajectories on the union $\mathcal{H}_1\cup\mathcal{H}_2$ with the local configuration shown in figure \ref{fig:14} (e). This technique of combining local trajectories is used in the following proof of lemma \ref{thm:noncross} to show that the local self-avoiding property can be extended to the entire honeycomb lattice.
\end{example}

Before we prove the following lemma, we introduce the notion of a crossing. We say that the particle \emph{crosses} its trajectory at time $\tau>0$ if $\mathbf{r}(\tau)=\mathbf{r}(t)$ for some $t<\tau$. We now show that if a flipping rotator system has an admissible initial configuration then the particle's trajectory is a self-avoiding walk (non-crossing), up to the first time it returns to its initial position.

\begin{lemma}\label{thm:noncross}
Suppose the initial configuration $C$ is admissible. If $\tau>0$ is the first time the particle crosses its trajectory in the system $(H,I,C)$ then $\mathbf{r}(\tau)=\mathbf{r}(0)$. Moreover, $\gamma=\{\mathbf{r}(t):0\leq t\leq\tau\}$ is a local cycle.
\end{lemma}

\begin{proof}
Let $\Omega\subset H$ be a finite union of hexagons, $C$ an admissible configuration, and let $\chi_1:\mathcal{L}_{\Omega}(C)\rightarrow P$ be a triperfect partition of $\mathcal{L}_{\Omega}(C)$. Additionally let, $\mathcal{H}$ be a single hexagon of $H$ where the lattice sites of $\Omega\cap\mathcal{H}=\{\mathbf{h}_1,\dots,\mathbf{h}_j\}$ have the property that $\mathbf{h}_i$ is adjacent to $\mathbf{h}_{i+1}$ for each $1\leq i<j$. Since $\mathcal{H}$ is a single hexagon and $C$ is admissible, lemma \ref{lem:tri} implies the set $\mathcal{L}_{\mathcal{H}}(C)$ has a triperfect partition, which we denote by $\chi_2:\mathcal{L}_{\mathcal{H}}(C)\rightarrow P$.

The claim is that the union $\Omega\cap\mathcal{H}$ has a triperfect partition. To verify this, note that properties \ref{property:1}-\ref{property:4} of $\chi_1$ and $\chi_2$ imply that there is a single local trajectory $\ell_{ab}(\Omega)\in\mathcal{L}_\Omega(C)$ and a single local trajectory $\ell_{ab}(\mathcal{H})\in\mathcal{L}_\Omega(\mathcal{H})$, both of which first passe through $\mathbf{h}_a$ and then through $\mathbf{h}_b$, if $\mathbf{h}_a,\mathbf{h}_b\in\Omega\cap\mathcal{H}$ are adjacent. Since $\mathbf{h}_1,\mathbf{h}_2\in\Omega\cap\mathcal{H}$ are adjacent, without loss in generality, we may assume that $\ell_{12}(\Omega)\in\mathcal{L}^1_{\Omega}(C)$, $\ell_{12}(\mathcal{H})\in\mathcal{L}^1_{\mathcal{H}}(C)$ and that $\ell_{21}(\Omega)\in\mathcal{L}^2_{\Omega}(C)$, $\ell_{21}(\mathcal{H})\in\mathcal{L}^2_{\mathcal{H}}(C)$.

Now suppose, $\Psi\subset H$ is the two lattice sites $\mathbf{h}_1,\mathbf{h}_2\in\mathbb{H}$ together with the bond between them. As property \ref{property:4} holds at the lattice sites $\mathbf{h}_1,\mathbf{h}_2\in\mathbb{H}$ for both partitions $\chi_1$ and $\chi_2$ then these partitions agree on $\Psi$. Furthermore, we observe that property \ref{property:4} implies that if $\chi_1$ and $\chi_2$ agree at the lattice site $\mathbf{h}_i\in\Omega\cap\mathcal{H}$ then $\chi_1$ and $\chi_2$ agree at the lattice site $\mathbf{h}_{i+1}\in\Omega\cap\mathcal{H}$ adjacent to $\mathbf{h}_i$. Hence, the partitions $\chi_1$ and $\chi_2$ agree on the intersection $\Omega\cap\mathcal{H}$.

Suppose then that $\ell=\{\mathbf{r}(t_1),\dots,\mathbf{r}(t_2)\}\in\mathcal{L}_{\Omega\cup\mathcal{H}}(C)$ and that $\ell\not\subset\Omega,\mathcal{H}$. Hence, there must be a time $\tau$ where $t_1\leq\tau\leq t_2$ at which the particle either exits or enters $\Omega$. If $\ell$ is a local cycle then, as $\ell$ is baseless, we may assume that $\tau=t_1$. Under this assumption, we let $\{T_i\}_{1<i\leq n}$ denote the set of times at which the particle either enters or exits the set $\Omega$ for $t_1<T_i<t_2$. We also order these times so that $T_{i}<T_{i+1}$ and set $T_1=t_1$ and $T_{n+1}=t_2$. Using the notation, \begin{equation}\label{def:dec}
\ell_i=\{\mathbf{r}(T_i),\dots,\mathbf{r}(T_{i+1})\} \ \text{for} \ 1\leq i\leq n
\end{equation}
it follows that $\ell_i\in\mathcal{L}_{\Omega}(C)$ if and only if $\ell_{i+1}\in\mathcal{L}_{\mathcal{H}}(C)$ for each $1\leq i<n$. If $\ell$ is a local crossing, it can similarly be decomposed into the union $\ell=\cup_{i=1}^n\ell_i$ where $\ell_i$ is given by equation \eqref{def:dec} and $\ell_i\in\mathcal{L}_{\Omega}(C)$ if and only if $\ell_{i+1}\in\mathcal{L}_{\mathcal{H}}(C)$ for each $1\leq i<n$. In the case that $\ell\subset\Omega,\mathcal{H}$ then $\ell$ is either a local trajectory of $\mathcal{L}_{\Omega}(C)$ or $\mathcal{L}_{\mathcal{H}}(C)$. Hence, each $\ell\in\mathcal{L}_{\Omega\cup\mathcal{H}}(C)$ can be written as the union $\ell=\cup_{i=1}^n\ell_i$ of an alternating sequence of local trajectories from $\mathcal{L}_{\Omega}(C)$ and $\mathcal{L}_{\mathcal{H}}(C)$.

For the local trajectories $\ell_i$, that make up $\ell=\cup_{i=1}^n\ell_i$, let the function
$$\chi^*(\ell_i)=
\begin{cases}
\chi_1(\ell_i) \ \text{if} \ \ell_i\in\mathcal{L}_{\Omega}(C)\\
\chi_2(\ell_i) \ \text{if} \ \ell_i\in\mathcal{L}_{\mathcal{H}}(C)
\end{cases}.$$
Suppose then that $\chi^*(\ell_m)=P_k$ for some $1\leq m\leq n$ and some $1\leq k\leq 3$. Since the partitions $\chi_1$ and $\chi_2$ agree on $\mathbf{r}(T_i)\in\Omega\cap\mathcal{H}$ for each $1<i\leq n$ then $\chi^*(\ell_i)=P_k$ for each $1\leq i\leq n$. That is, $\chi_1$ and $\chi_2$ assign each $\ell_i\subseteq\ell$ the same element of $P$.

To use this, let $\chi:\mathcal{L}_{\Omega\cup\mathcal{H}}(C)\rightarrow P$ be the partition given by $\chi(\ell)=P_k$ if $\chi^*(\ell_1)=P_k$ where $\ell\in\mathcal{L}_{\Omega\cup\mathcal{H}}(C)$ has the decomposition $\cup_{i=1}^n\ell_i$ and $\ell_i$ is given by equation \eqref{def:dec}. The partition $\chi$ is then well defined on any local crossing $\ell\in\mathcal{L}_{\Omega\cup\mathcal{H}}(C)$ since $\ell_1$ is uniquely defined in this case. If $\ell\in\mathcal{L}_{\Omega\cup\mathcal{H}}(C)$ is a local cycle, $\ell_1$ is not uniquely defined, but $\chi$ is still well defined on $\ell$ since $\chi^*(\ell_i)=\chi^*(\ell_1)$ for any $1\leq i\leq n$, no matter where the base of $\ell$ is chosen to be.

The claim then, is that $\chi$ is a triperfect partition of $\mathcal{L}_{\Omega\cup\mathcal{H}}(C)$. To verify that $\chi$ has property \ref{property:1}, suppose to the contrary that $\ell=\cup_{i=1}^n\ell_i\in\mathcal{L}^k_{\Omega\cup\mathcal{H}}(C)$ is a local crossing that is not self-avoiding. Since each $\ell_i\subset\ell$ cannot cross itself, there are two $\ell_p,\ell_q\subset\ell$ where $p<q$ and $\ell_p\cap\ell_q\neq\emptyset$.
Moreover, either $\ell_p,\ell_q\in\mathcal{L}_{\Omega}^k(C)$ or $\ell_p,\ell_q\in\mathcal{L}_{\mathcal{H}}^k(C)$. Given that property \ref{property:4} holds for both $\chi_1$ and $\chi_2$ then $\ell_p=\ell_q$. Additionally, since there is a first time $\tau>t_1$ that $\ell$ crosses itself we may assume that at this point in time $\mathbf{r}(\tau)\in\ell_q$, which implies that $\tau=T_q$. Assuming, without loss in generality, that $\ell_p,\ell_q\in\mathcal{L}_{\Omega}(C)$ then $\mathbf{r}(T_p)=\mathbf{r}(T_q)$ and $\mathbf{r}(T_p+1)=\mathbf{r}(T_q+1)$ since these must both be lattice sites of $\Omega$. Therefore, $\mathbf{r}(T_p)=\mathbf{r}(T_q)$, $\mathbf{v}(T_p)=\mathbf{v}(T_q)$, and $C_{\mathbf{r}(T_p+1)}(T_p)=-C_{\mathbf{r}(T_q+1)}(T_q)$ since $\tau=T_q$ is the first time $\ell$ crosses itself. However, this implies that $\mathbf{r}(T_p+2)\neq\mathbf{r}(T_q+2)$, using equations $\eqref{eq:1}-\eqref{eq:3}$, implying that $\ell_p\neq\ell_q$ since at least one of $\mathbf{r}(T_p+2)$ or $\mathbf{r}(T_q+2)$ is in $\Omega$. Thus, there is no first time $\tau>t_1$ at which $\ell$ crosses itself, so $\chi$ has property \ref{property:1}.

To verify that $\chi$ has property \ref{property:4} let $\ell_*$ be a local trajectory of either $\mathcal{L}_{\Omega}(C)$ or $\mathcal{L}_{\mathcal{H}}(C)$. Using properties \ref{property:1}-\ref{property:4} of $\chi_1$ and $\chi_2$, and the fact that $\Omega\cup\mathcal{H}$ is finite, one can show that $\ell_*\subseteq\ell$ where $\ell\in\mathcal{L}_{\Omega\cup\mathcal{H}}(C)$. Moreover, if $\chi^*(\ell_*)=P_k$ then $\chi(\ell)=P_k$. Since $\ell$ must be the unique local trajectory containing $\ell_*$ then the fact that both $\chi_1$ and $\chi_2$ have property \ref{property:4} implies that $\chi$ also has property \ref{property:4}. Hence, $\chi:\mathcal{L}_{\Omega\cup\mathcal{H}}(C)\rightarrow P$ is a triperfect partition.

To show there is a triperfect partition of $\mathcal{L}_{H}(C)$ we let $\Omega_1$ be any one of the hexagons containing the origin, which has a triperfect partition by lemma \ref{lem:tri}. We then let $\Omega_2$ be the ring of hexagons that share a lattice site with $\Omega_1$. By merging individual hexagons of $\Omega_2$ with $\Omega_1$ we can obtain a triperfect partition of $\Omega_1\cup\Omega_2$. Continuing in this manner, if $\Omega_m$ is the ring of hexagons surrounding $\cup_{i=1}^{m-1}\Omega_{i}$ then it follows, by induction, that there is a triperfect partition $\hat{\chi}:\mathcal{L}_H(C)\rightarrow P$ where $\cup_{i=1}^\infty\Omega_{i}=H$.

To use the partition $\hat{\chi}$, suppose $\tau>0$ is the first time the particle crosses its trajectory in the $(H,I,C)$ system, where $I=(\mathbf{r},\mathbf{v})$. Since the partition $\hat{\chi}$ has property \ref{property:4}, there is exactly one local trajectory $\ell\in\mathcal{L}_H(C)$ that first passes through $\mathbf{r}$ then $\mathbf{r}+\mathbf{v}$. Since the sequence of positions $\{\mathbf{r}(t):0\leq t\leq \tau\}\subseteq\ell$ then $\ell=\{\mathbf{r}(t):0\leq t\leq \tau\}$, which must be a local cycle of $\mathcal{L}_{H}(C)$.
\end{proof}

Lemma \ref{thm:noncross} states that a particle cannot cross its trajectory away from its initial position if the initial configuration is admissible. The natural question then is, what happens if the particle does return to its initial position. In the following theorem we prove that, at the time the particle returns to its initial position, the configuration of scatterers is again admissible. Therefore, this process repeats itself so that, between returns to its initial position, the particle's trajectory is always a self-avoiding walk, no matter how many times it returns. Additionally, if the particle has a last time at which it returns to its initial position its trajectory thereafter is a self-avoiding walk.

\begin{theorem}\label{lem:1}\textbf{(Self-Avoiding Motion)}
Let $(H,I,C)$ be a system with the admissible configuration $C$.\\
(a) If the particle has the sequence of return times $\{\tau_i\}_{i\geq 0}$ to its initial position $\mathbf{r}$ then $\gamma_i=\{\mathbf{r}(t):\tau_{i-1}\leq t\leq \tau_i\}$
is a cycle based at $\mathbf{r}$ for each $i\geq 1$.\\
(b) If $\tau$ is the last time $\mathbf{r}(\tau)=\mathbf{r}$ then $\{\mathbf{r}(t)\}_{t\geq \tau}$ is a sequence of distinct lattice sites.
\end{theorem}

\begin{proof}
In the system $(H,I,C)$ let $\tau_1>0$ be the particle's first return time to its initial position $\mathbf{r}$. Assuming $C$ is admissible then lemma \ref{thm:noncross} implies that for $0\leq t\leq \tau_1$ the particle moves on the local cycle
$\gamma_1=\{\mathbf{r}(t):0\leq t\leq \tau_1\}$ based at the particle's initial position $\mathbf{r}$. Additionally, theorem \ref{thm:partition} implies that every hexagon $\mathcal{H}\subset H$, which the particle has either (i) first entered then exited or (ii) not visited by time $\tau_1$, has a local admissible configuration at time $\tau_1$. Thus, the only hexagons that may not have an admissible configuration at time $\tau_1$ are those that contain the particle's initial position $\mathbf{r}$.

Observe that there are three hexagons in $\mathcal{H}_i\subset H$ for $i=1,2,3$ that contain $\mathbf{r}$. If the particle does not exit $\mathcal{H}_i$ for some $1\leq i\leq 3$ by the time $\tau_1$, then $\tau_1=6$ and one can check that $C(\tau_1)$ is admissible. Suppose then that $\mathbf{r}(t_1)\notin\mathcal{H}_i$ for some $1\leq i\leq 3$ at some time $0<t_1<\tau_1$. Under this assumption, let $(H,J,C)$ be the system with the initial condition $J=(\mathbf{r}(t_1),\mathbf{v}(t_1))$ and position given by $\mathbf{s}(t)$. Then $\gamma_1$ is equal to the cycle $\{\mathbf{s}(t):0\leq t\leq\tau_1\}$ as a set of positions since $\gamma_1$ is a local cycle and therefore baseless. Moreover, the configuration in the systems $(H,I,C)$ and $(H,J,C)$ are the same at time $\tau_1$ since the particle has moved over the same set of lattice sites up to this point in time. However, in the $(H,J,C)$ system the particle has first entered then exited the hexagon $\mathcal{H}_i$ so that this hexagon has an admissible configuration at time $\tau_1$, according to theorem \ref{thm:partition}. Therefore, each hexagon of $H$ in the $(H,I,C)$ system has an admissible configuration at time $\tau_1$.

Continuing inductively, repeated use of lemma \ref{thm:noncross} implies that  if $\{\tau_i\}_{i\geq 0}$ is a sequence of return times in the $(H,I,C)$ system, then each sequence of positions $\gamma_i=\{\mathbf{r}(t):\tau_{i-1}\leq t\leq\tau_i\}$ is a local cycle. Moreover, if $\tau$ is the last time $\mathbf{r}(\tau)=\mathbf{r}$ then, as $C(\tau)$ is admissible, lemma \ref{thm:noncross} implies that $\{\mathbf{r}(t)\}_{t\geq \tau}$ is a sequence of distinct lattice sites.
\end{proof}

Before finishing this section, we note that an admissible configuration is defined locally in the sense that a configuration $C$ is admissible if its restriction to each hexagon of $H$ is admissible. In contrast, the self-avoiding behavior described in theorem \ref{lem:1} is not restricted to a particular hexagon of the lattice. In fact, in a system with an admissible configuration, the only way the particle can cross its trajectory is by returning to its initial position, no matter how far the particle moves from this lattice site. That is, although the notion of an admissible configuration is defined locally it has a nonlocal effect on the particle's trajectory.

We also note that the self-avoiding motion, which we find in a system $(H,I,C)$ with an admissible configuration $C$, is fundamentally different from using a probabilistic rule to generate a self-avoiding walk. If one uses a probabilistic rule to create a self-avoiding walk then there is always a need to modify this rule whenever the particle's trajectory threatens to cross itself. In contrast, the particle's motion in the $(H,I,C)$ system is completely deterministic, so that no provisional rules are needed to keep it from crossing its trajectory, at least away from its initial position.

Related to this, one can also consider placing probabilistic scatters on the honeycomb lattice, which rotate the particle to the right with some probability $p$ and to the left with probability $1-p$. If a few of these scatterers are placed near the origin in the $(H;\mathbf{R})$ model, what we observe numerically is that the particle will eventually cross its trajectory away from the origin. That is, even a small amount of randomness in the $(H;\mathbf{R})$ can destroy the particle's self-avoiding behavior.

\section{Blocking Configurations}\label{sec:7}
In this section we introduce the notion of a blocking configuration of scatterers. What we show is that if the configuration $C$ is both an admissible and a blocking configuration, then the particle in the $(H,I,C)$ system will return to its initial position an infinite number of times. This will allow us to prove, as stated in theorem \ref{thm:1}, that there are an infinite number of times $\{\tau_i\}_{i\geq 0}$ at which the particle returns to its initial position in the $(H;\mathbf{R})$ model.

\begin{definition}\label{def:blocking}
Let $C$ be a configuration of scatterers on $H$ and $\tau_b<\infty$ a finite time. Suppose that for any initial condition $I$, there is a time $t_I\leq\tau_b$ where $\mathbf{r}(t_I)=\mathbf{r}$ in the $(H,I,C)$ system. Then $C$ is called a \emph{blocking configuration} and $\tau_b$ a \emph{blocking time} of $C$.
\end{definition}

Essentially, the configuration $C$ is a blocking configuration if, no matter where the particle starts on the honeycomb lattice $H$, the particle always returns to this initial position within $\tau_b$ time steps. As an example, the initial configuration of all right scatterers is a blocking configuration. The reason for this is that for any initial condition, i.e. initial position and velocity, on the honeycomb lattice the particle will return to its initial position after six time steps. Hence, the time $\tau_b=6$ is a blocking time for the configuration of all right scatterers.

The reason we are interested in blocking configurations is that they have the following property. In a system $(H,I,C)$ with a blocking configuration $C$, the particle can only visit $\tau_b-1$ new lattice sites before having to return to a site it has previously visited. Hence, the particle's progress through the lattice is temporarily ``blocked" by this configuration of scatterers. Blocking configurations do not indefinitely stop the progress of a particle through the lattice. The reason is that, as time increases, the region the particle has visited also increases so that the particle gets blocked further and further away from its initial position.

With this in mind, an important concept in this context is the collection of lattice sites a particle has visited by time $t\geq0$, and those it has not. For a system $(H,I,C)$ we let $V(t)$ denote the set of lattice sites that have been visited by the particle by time $t\geq0$, so that $V(t)=\{\mathbf{r}(i): 0\leq i\leq t\}$.
Similarly, we let $U(t)$ denote the collection of lattice sites the particle has not visited by time $t$, so that $U(t)=\mathbb{H}-V(t)$.

The following theorem states that, if $C$ is both an admissible and a blocking configuration then the particle in the system $(H,I,C)$ will have an infinite number of times at which it returns to its initial position. Moreover, between these returns the particle's trajectory will be a self-avoiding walk.

\begin{theorem}\label{thm:recurrence}\textbf{(Recurrence Condition)}
Suppose $C$ is an admissible blocking configuration. Then, for any initial condition $I=(\mathbf{r},\mathbf{v})$, the $(H,I,C)$ system has\\
(a) an infinite sequence of times $\{\tau_i\}_{i\geq0}$ at which $\mathbf{r}(\tau_i)=\mathbf{r}$; and\\
(b) for each $i\geq 1$, the particle moves on the cycle $\gamma_i=\big\{\mathbf{r}(t):\tau_{i-1}\leq t\leq\tau_{i}\big\}$.
\end{theorem}

\begin{proof}
Let $C$ be an admissible blocking configuration with blocking time $\tau_b<\infty$. By definition $\mathbf{r}(\tau_0)=\mathbf{r}$ for $\tau_0=0$ in the system $(H,I,C)$. Continuing by induction, suppose that $\tau_k$ is the particle's $k$-th return time to its initial position. By way of contradiction, we then suppose that $\tau_k$ is the last time the particle returns to its initial position. As $C$ is admissible, theorem \ref{lem:1} implies that $\{\mathbf{r}(t)\}_{t\geq \tau_k}$ is a distinct sequence of lattice sites.

Since the collection of sites $V(\tau_k)$ consists of a finite number of lattice sites, this implies that there is a first time $t_1>\tau_k$, such that $\mathbf{r}(t_1)\in V(\tau_k)$ and $\mathbf{r}(t_1+1)\in U(\tau_k)$. That is, there is a first time $t_1$ the particle exits $V(\tau_k)$. Furthermore, since $C$ is a blocking configuration, there is a first time $s_1>t_1$ that the particle enters $V(\tau_k)$ where $s_1-t_1\leq \tau_b$.

Once the particle has returned to $V(\tau_k)$, the fact that $\{\mathbf{r}(t)\}_{t\geq \tau_k}$ is a distinct sequence of lattice sites and $V(\tau_k)$ is finite, implies that there is a second time $t_2\geq s_1$ that the particle exits $V(\tau_k)$. Again, since $C$ is a blocking configuration and $\{\mathbf{r}(t)\}_{t\geq \tau_k}$, a distinct sequence of lattice sites, there is a second time $s_2>t_2$ at which the particle enters $V(\tau_k)$. Continuing in this manner, there is an infinite sequence of times $\{t_i\}_{i\geq 1}$ at which the particle exits $V(\tau_k)$. But this is not possible since $\{\mathbf{r}(t_i)\}_{i\geq 1}$ is an infinite set of lattice sites contained in the $V(\tau_k)$. Hence, there is a time $\tau_{k+1}>\tau_k$ at which $\mathbf{r}(\tau_{k+1})=\mathbf{r}$.

By induction it then follows that there is an infinite sequence of times $\{\tau_i\}_{i\geq 0}$ at which $\mathbf{r}(\tau_i)=\mathbf{r}$. Part (b) of theorem \ref{thm:recurrence} follows from theorem \ref{lem:1} since the configuration $C$ is assumed to be admissible.
\end{proof}

With theorem \ref{thm:recurrence} in place we are now in a position to give a proof of theorem \ref{thm:1}. The reason we are able to give a proof here, is that the initial configuration of all right scatterers in the model $(H;\mathbf{R})$ is both an admissible and a blocking configuration. Hence, theorem \ref{thm:recurrence} immediately implies part (a) of theorem \ref{thm:1}. In order to prove part (b), however, we will need to use specific properties of this model's initial configuration. For convenience, in the proof we let $\underline{\mathbf{x}}$ be the reflection of $\mathbf{x}\in\mathbb{R}^2$ over the $x$-axis and $\bar{\mathbf{x}}$ the reflection of $\mathbf{x}$ over the line $x=1/2$, respectively. A proof of theorem \ref{thm:1} is the following.

\begin{proof}
To begin, we note that part (a) of theorem \ref{thm:1} is a direct consequence of theorem \ref{thm:recurrence}. To prove part (b) we let $C$ denote the configuration of all right scatterers and let $\tau_i$ be the particle's $i$-th return time to the origin. Since $C$ is admissible, the proof of theorem \ref{lem:1} implies that each cycle $$\gamma_i=\{\mathbf{r}(t):\tau_{i-1}\leq t\leq\tau_i\}$$
is a local cycle. Given that $\mathbf{v}(\tau_0)=(1,0)$ at time $\tau_0=0$, the fact that each $\gamma_i$ is a local cycle implies that $\mathbf{v}(\tau_i)=(1,0)$ for each $i\geq 0$.

Let $k\geq 0$ be a fixed integer. Given that $\mathbf{r}(\tau_k)=(0,0)$ and $\mathbf{v}(\tau_k)=(1,0)$, then between time $t=\tau_k$ and time $t=\tau_k+1$ the particle crosses the line $x=1/2$. In order for the particle to return to the origin there must be some first time $\tau_k+\tau$ such that between time $\tau_k+\tau$ and $\tau_k+\tau+1$ the particle again cross the line $x=1/2$. The claim is that, if the configuration $C(\tau_k)$ is symmetric with respect to the line $x=1/2$ then
\begin{equation}\label{eq:10}
\mathbf{r}(\tau_k+\tau+t)=\overline{\mathbf{r}}(\tau_k+\tau-t+1) \ \text{for} \ 1\leq t\leq\tau.
\end{equation}
To verify that \eqref{eq:10} holds under this assumption, we let
$\mathbf{s}(t)=\mathbf{r}(\tau_k+\tau+t)$, $\mathbf{w}(t)=\mathbf{v}(\tau_k+\tau+t)$, and $S(t)=C(\tau_k+\tau+t)$. Note that at time $t=1$ we have $\mathbf{s}(1)=\bar{\mathbf{s}}(0)$ since the particle crosses the line $x=1/2$ between time $\tau_k+\tau$ and $\tau_k+\tau+1$.

Proceeding by induction, we assume that
\begin{equation}\label{eq:assump}
\mathbf{s}(T-1)=\bar{\mathbf{s}}(-T+2) \ \text{and} \ \mathbf{s}(T)=\bar{\mathbf{s}}(-T+1)
\end{equation}
for some fixed time $T$, where $1\leq T<\tau$. Then $\mathbf{w}(T-1)=-\underline{\mathbf{w}}(-T+1)$. Using the equations of motion \eqref{eq:1}-\eqref{eq:3} it follows that
\begin{gather}
\mathbf{w}(T)=R\big[S_{\mathbf{s}(T)}(T-1)\big]\mathbf{w}(T-1) \ \text{and}\\
\mathbf{w}(-T)=R\big[-S_{\mathbf{s}(-T+1)}(-T)\big]\mathbf{w}(-T+1).
\end{gather}
To simplify our notation, let $K\in\mathbb{R}^{2\times 2}$ denote the matrix that reflects each $\mathbf{x}\in\mathbb{R}^2$ over the $x$-axis. Then for each $z\in\{-1,1\}$ and $\mathbf{x}\in\mathbb{R}^2$ the transformation
\begin{equation}
R(-z)\underline{\mathbf{x}}=R(-z)K\mathbf{x}=
\left[\begin{array}{cc}
-\cos(\frac{\pi}{3}z)&\sin(\frac{\pi}{3}z)\\
\sin(\frac{\pi}{3}z)&\cos(\frac{\pi}{3}z)
\end{array}\right]\mathbf{x}=
KR(z)\mathbf{x}=\underline{R(z)\mathbf{x}}.
\end{equation}
Additionally, under the assumption that $S=C(\tau_k)$ is symmetric with respect to the line $x=1/2$ and is admissible, then the self-avoiding property implies that $S_{\mathbf{s}(T)}(T-1)=S_{\mathbf{s}(-T+1)}(-T)$. Hence,
\begin{align}
\underline{\mathbf{w}}(T)&=\underline{R[S_{\mathbf{s}(T)}(T-1)]\mathbf{w}(T-1)}= R[-S_{\mathbf{s}(-T+1)}(-T)]\underline{\mathbf{w}}(T-1)\\
&=-R[-S_{\mathbf{s}(-T+1)}(-T)]\mathbf{w}(-T+1)=\mathbf{w}(-T).
\end{align}
Using the fact that $\overline{\mathbf{x}+\mathbf{y}}=\overline{\mathbf{x}}-\underline{\mathbf{y}}$ for $\mathbf{x},\mathbf{y}\in\mathbf{R}^{2}$ it follows that
\begin{align}
\bar{\mathbf{s}}(T+1)&=\overline{\mathbf{s}(T)+\mathbf{w}(T)}= \bar{\mathbf{s}}(T)-\underline{\mathbf{w}}(T)\\
&=\bar{\mathbf{s}}(T)-\mathbf{w}(-T)=\mathbf{s}(-T+1)-\mathbf{w}(-T)=\mathbf{s}(-T).
\end{align}
Therefore, the assumption in equation \eqref{eq:assump} implies $\mathbf{s}(T+1)=\bar{\mathbf{s}}(-T)$. As equation \eqref{eq:assump} holds for $T=1$ then, by induction, it follows that $\mathbf{s}(t-1)=\bar{\mathbf{s}}(-t+2)$ for all $1\leq t\leq \tau$. Since $\mathbf{r}(\tau_k+2\tau)=\mathbf{s}(\tau)=\bar{\mathbf{s}}(-\tau+1)= \bar{\mathbf{r}}(\tau_k+1)= \mathbf{r}(\tau_k)$
then
\begin{equation*}
\gamma_k=\{\mathbf{r}(t):\tau_k\leq t\leq\tau_k+2\tau\}
\end{equation*}
is a cycle based at the origin, which is symmetric with respect to the line $x=1/2$. Therefore, from time $t=\tau_k$ to time $t=\tau_{k+1}=\tau_k+2\tau$ the particle has flipped $2\tau$ scatterers, whose positions as a set are symmetric with respect to the line $x=1/2$.

Under the assumption that $C(\tau_k)$ is symmetric with respect to this line, this implies that the configuration of scatterers $C(\tau_{k+1})$ at the particle's $k+1$-st return to the origin is also symmetric with respect to the line $x=1/2$. Moreover, since the initial configuration $C(\tau_0)=C$ is symmetric with respect to $x=1/2$ at time $\tau_0=0$, it then follows by induction that $C(\tau_i)$ is also symmetric with respect to the line $x=1/2$ for all $i\geq 0$. Therefore, for $i>0$ each
\begin{equation*}
\gamma_i=\{\mathbf{r}(t):\tau_{i-1}\leq t\leq\tau_i\}
\end{equation*}
is a cycle based at the origin, which is symmetric with respect to the line $x=1/2$. This completes the proof.
\end{proof}

\begin{figure}
    \begin{overpic}[scale=.2]{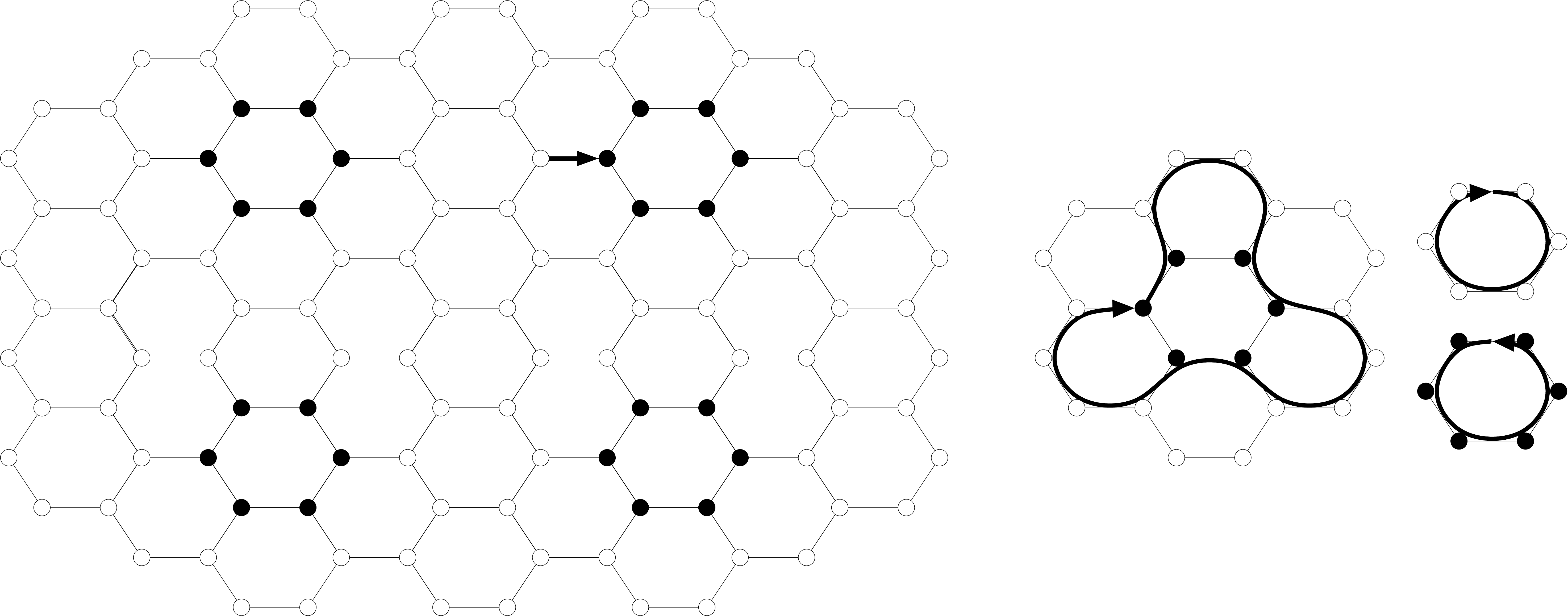}
    \put(76,13){$\Gamma_1$}
    \put(94,8){$\Gamma_3$}
    \put(94,28.5){$\Gamma_2$}
    \put(24.5,-3.5){(a) $A$}
    \put(31.5,28.5){$\mathbf{r}$}
    \put(35.5,30){$\mathbf{v}$}
    \put(85,5){(b)}
    \end{overpic}
    \vspace{0.1in}
\caption{The admissible blocking configuration $A$ considered in example \ref{ex:3} is shown in (a). The configuration's possible first return cycles $\Gamma_1$, $\Gamma_2$, and $\Gamma_3$, up to rotation and translation, are shown in (b).}\label{fig:15}
\end{figure}
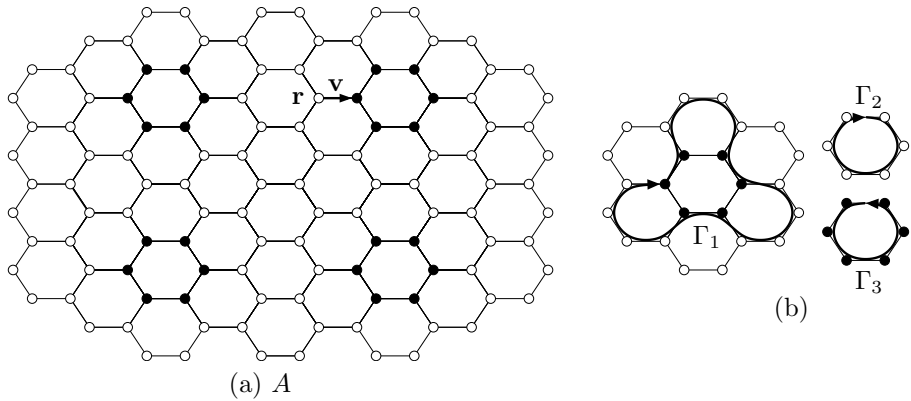

If $C$ is an admissible blocking configuration, theorem \ref{thm:recurrence} only guarantees that the particle in the $(H,I,C)$ system moves on a sequence of cycles. It does not guarantee that these cycles are symmetric with respect to the line $x=1/2$, as is the case in the $(H;\mathbf{R})$ model. To illustrate this fact we consider the following example, in which we study an admissible blocking configuration that is our first example of a configuration that does not consist of all right scatterers.

\begin{example}\label{ex:3}
Consider the configuration $A$ shown in figure \ref{fig:15} (a), which is the configuration of all right scatterers with periodically spaced hexagons of left scatterers. This configuration is admissible since each hexagon on the lattice has either the configuration $(i)$,$(ii)$, or $(vii)$, each of which is an element of $\mathcal{A}$ (see figure \ref{fig:8}).

The claim is that $A$ is also a blocking configuration To check whether this is the case we need to show that, for each initial condition $I=(\mathbf{r},\mathbf{v})$, the particle has a first return time to its initial position $\mathbf{r}$ that is bounded by a fixed number $\tau_b<\infty$. To verify this we observe, for a given initial condition $I$, that up to translation and rotation, the particle moves along one of three possible cycles $\Gamma_1$, $\Gamma_2$, and $\Gamma_3$ as shown in figure \ref{fig:15} (b). Since these cycles take between 8 and 16 time steps, this implies that $A$ is a blocking configuration with blocking time $\tau_b=16$.

Since $A$ is an admissible as well as a blocking configuration, it follows from theorem \ref{thm:recurrence} that the particle in the $(H;I,A)$ model has an infinite sequence of returns $\{\tau_i\}_{i\geq 0}$ to its initial position, where the $i$-th return is along the cycle $$\gamma_i=\{\mathbf{r}(t):\tau_{i-1}\leq t\leq \tau_i\}.$$ For the initial condition $I$ indicated by the arrow in figure \ref{fig:15}, a few of the particle's cycles are shown in figure \ref{fig:16}. We note that, unlike the cycles in the $(H;\mathbf{R})$ model, these cycles are not symmetric with respect to any line with the exception of $\gamma_3$. That is, the initial configuration of all right scatterers is a special example of an admissible blocking condition, which causes the particle to move along symmetric cycles. In general, an admissible blocking configuration does not lead to any symmetry in the particle's trajectory.

Despite this, there are a number of notable similarities between the motion of the particle in the $(H;\mathbf{R})$ and $(H;I,A)$ models. Besides the fact that both particles travel along an infinite sequence of cycles, based at their respective initial positions, both particles have a qualitatively similar time-averaged mean square displacement.

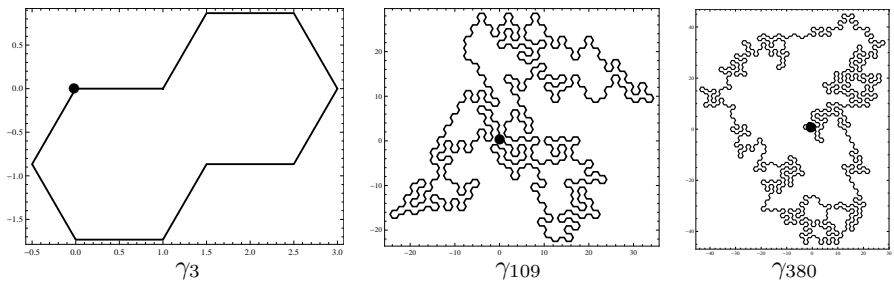
\begin{figure}
\begin{center}
\begin{tabular}{ccc}
    \begin{overpic}[scale=.49]{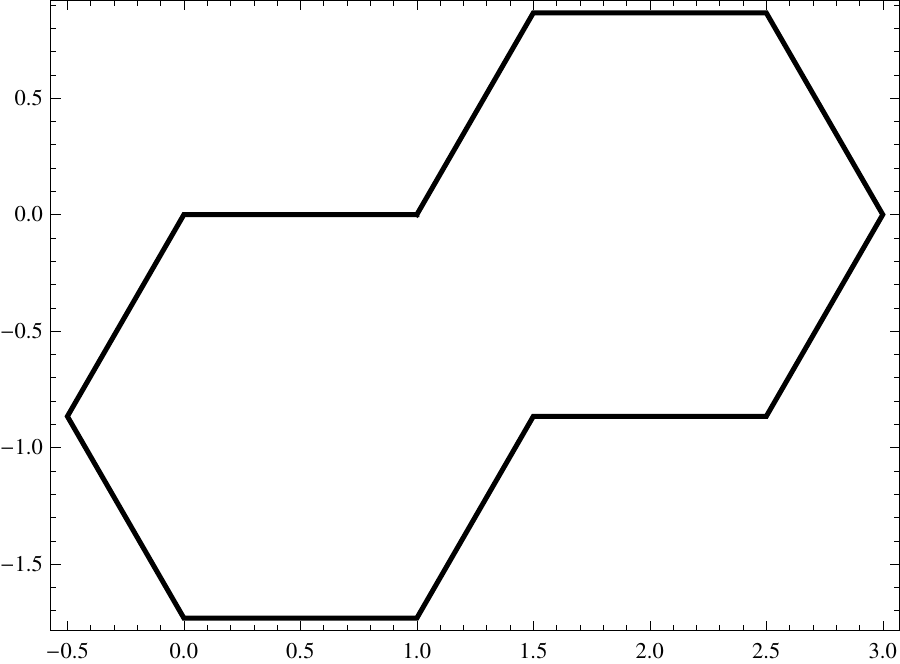}
    \put(50,-5){$\gamma_3$}
    \put(18,48){$\bullet$}
    \put(145,-5){$\gamma_{109}$}
    \put(227,-5){$\gamma_{380}$}
    \end{overpic} &
    \begin{overpic}[scale=.42]{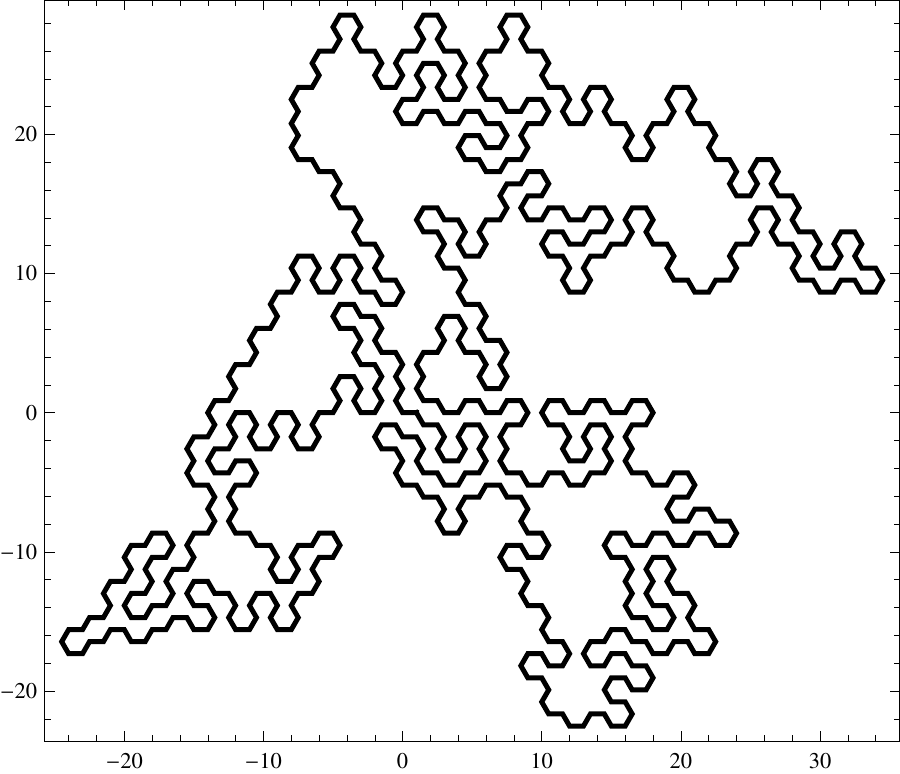}
    \put(42.5,38.5){$\bullet$}
    \end{overpic} &
    \begin{overpic}[scale=.30]{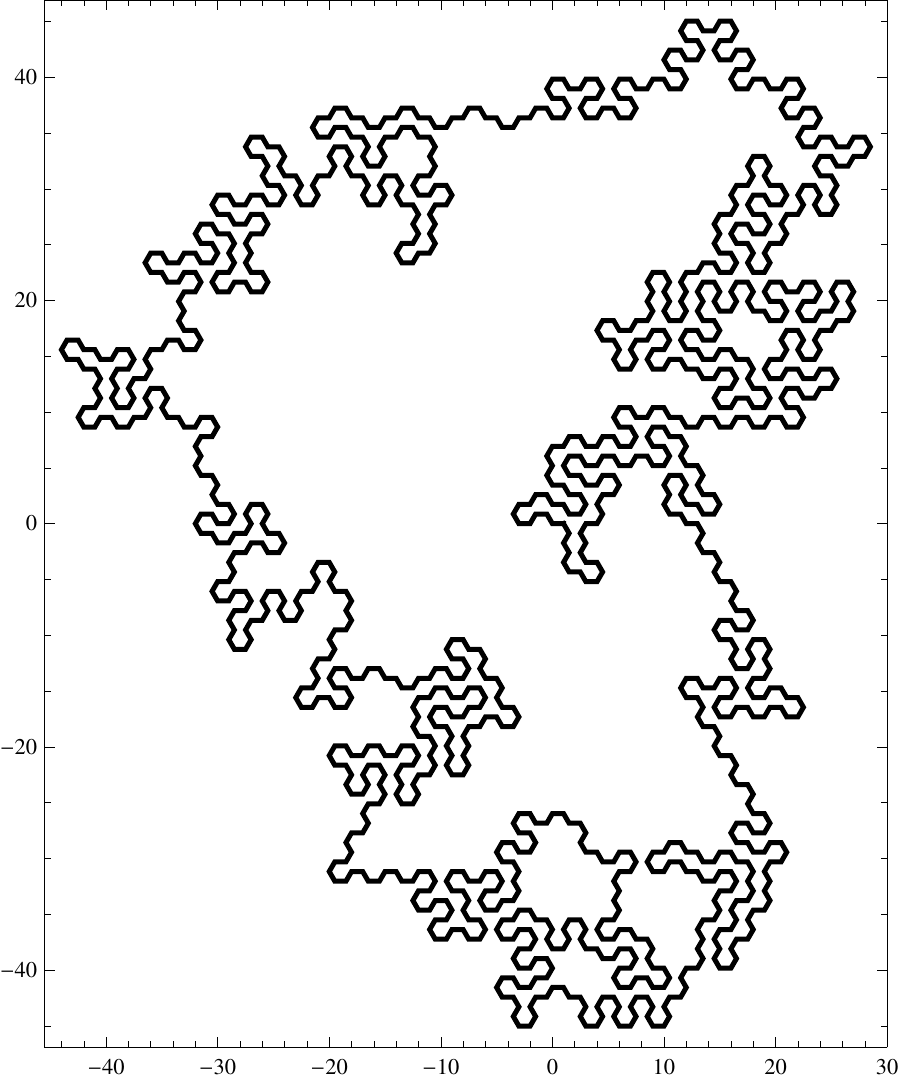}
    \put(48,49.5){$\bullet$}
    \end{overpic}
\end{tabular}
\end{center}
\caption{The figure shows the cycles $\gamma_3$, $\gamma_{190}$, and $\gamma_{380}$ that occur in the $(H;I,A)$ model in example \ref{ex:3}. Each cycle is based at the origin, which is indicated by a black dot.}\label{fig:16}
\end{figure}

The time-averaged mean square displacement of the particle in the $(H;I,A)$ model is shown in figure \ref{fig:17} (a), where $\bar{\triangle}(t)\approx \frac{4}{11}t^{8/13}$ for $t\leq 10^6$. Hence, up to time $t=10^6$ the particle displays time-averaged subdiffusion, similar to the particle in the $(H;\mathbf{R})$ model. We note that if the particle continues to subdiffuse in this manner for all time, then the particle has an unbounded trajectory and, therefore, has a pulsating motion. Additionally, the fraction of cycles $F(\ell)$ of length $\ell$ for $t\leq 10^6$ in the $(H;I,A)$ model are shown in figure \ref{fig:17} (b). Here, $F(\ell)\approx \frac{3}{5}t^{-3/2}$ so that, as in the $(H;\mathbf{R})$ model, short cycles dominate the particle's trajectory.
\end{example}

A natural question is whether the particle in the $(H,I,C)$ system will always exhibit time-averaged subdiffusion if $C$ is an admissible blocking configuration. Currently, it is unknown whether this holds in general although numerical simulations suggest that this is the case.

A second related question, is whether configurations exist that are admissible but do not have the blocking property. In the following example we show an example of such a configuration and describe its effect on the particle's motion.

\begin{example}\label{ex:4}
Consider the configuration $B$ shown in figure \ref{fig:18} (a), which has alternating layers of left and right scatterers. The configuration $B$ is admissible since its restriction to any hexagon is either the configuration $(i)$, $(ii)$, $(v)$, or $(vii)$ of $\mathcal{A}$ (see figure \ref{fig:8}).

\begin{figure}
\begin{center}
\begin{tabular}{cc}
    \begin{overpic}[scale=.66]{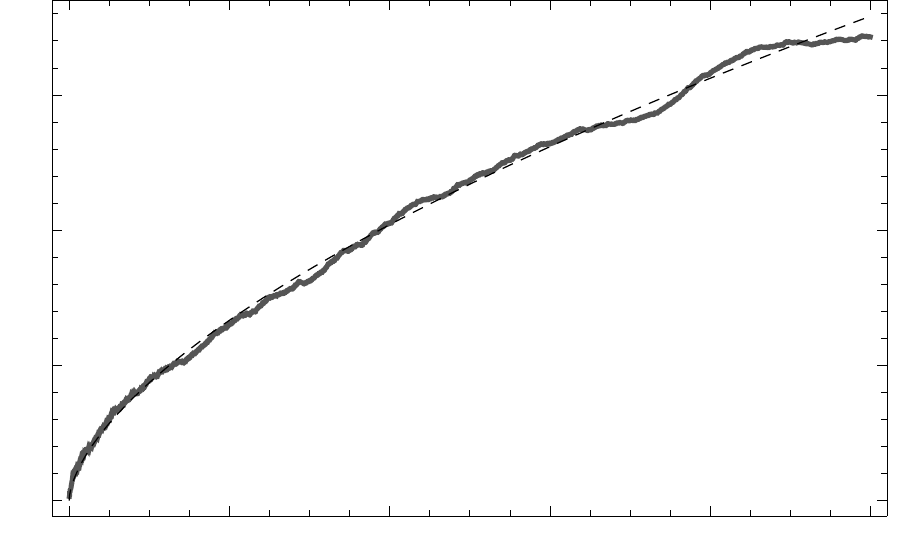}
    \put(46,-6.5){(a) \small$\bar{\triangle}(t)$}

    \put(5,-0.25){\tiny$0$}
    \put(21,-0.25){\tiny$2(10^5)$}
    \put(39,-0.25){\tiny$4(10^5)$}
    \put(57,-0.25){\tiny$6(10^5)$}
    \put(76,-0.25){\tiny$8(10^5)$}
    \put(94,-0.25){\tiny$10^6$}

    \put(2,2.75){\tiny$0$}
    \put(-1,19){\tiny$500$}
    \put(-3,34){\tiny$1000$}
    \put(-3,49){\tiny$1500$}
    \end{overpic} &
    \begin{overpic}[scale=.65]{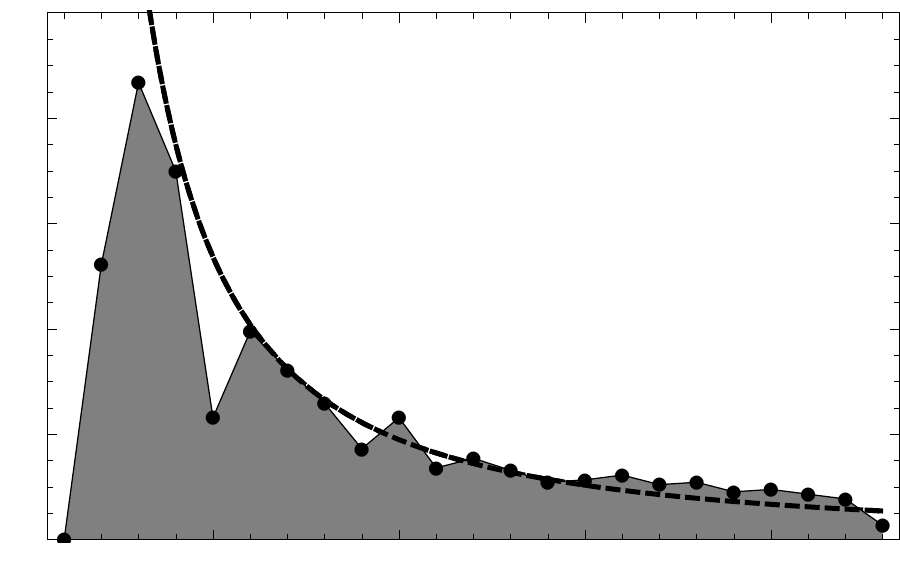}
    \put(5,0){\tiny$6$}
    \put(32,0){\tiny$42$}
    \put(64,0){\tiny$82$}
    \put(94,0){\tiny$122$}

    \put(2,2.75){\tiny$0$}
    \put(-1.25,14.25){\tiny$.02$}
    \put(-1.25,25.75){\tiny$.04$}
    \put(-1.25,37.25){\tiny$.06$}
    \put(-1.25,49.25){\tiny$.08$}
    \put(0,60.5){\tiny$.1$}

    \put(45,-6.5){(b) $F(\ell)$}
    \end{overpic}
\end{tabular}
\end{center}
\caption{The time-averaged mean square displacement $\bar{\triangle}(t)\approx \frac{4}{11}t^{8/13}$ and fraction of cycles $F(\ell)\approx \frac{3}{5}t^{-3/2}$ of length $\ell$ are shown for $t\leq 10^6$ for the model $(H;I,A)$ in example \ref{ex:3}}\label{fig:17}
\end{figure}
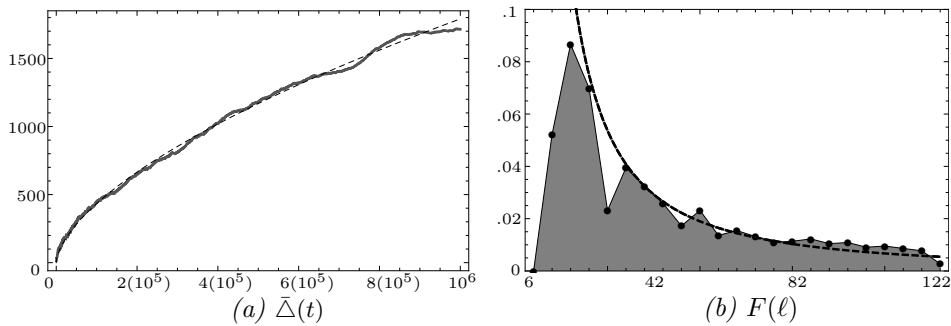

If $I=(\mathbf{r},\mathbf{v})$ is the particle's initial condition indicated in figure \ref{fig:18} (a), then the particle moves for all time to the right between the layers of left and right scatters in a zig-zag pattern, as shown in figure \ref{fig:18} (b). Since this motion persists indefinitely, the particle does not return to its initial position but escapes to infinity. The configuration $B$ is therefore an admissible configuration that is not a blocking configuration.

We note that since the particle in the model $(H;I,B)$ does not return to its initial position, the assumption that $C$ is a blocking configuration in theorem \ref{thm:recurrence} is necessary condition needed for this theorem to hold. This example also demonstrates that it is possible to have an admissible configuration $C$, such that the particle in the system $(H,I,C)$ never returns to its initial position, in which case the particle's entire trajectory is a self-avoiding walk, according to theorem \ref{lem:1}.
\end{example}

\section{Conclusion}\label{sec:conc}
In this paper we consider the motion of a particle in an LLG on the honeycomb lattice, where the particle is scattered either by flipping rotators or flipping mirrors. In both cases we find a new type of motion in which the particle returns to its initial position an infinite number of times and where, between these returns, the particle's trajectory is a self-avoiding walk. Our main example of this type of dynamics is the motion of the particle in the flipping rotator model $(H;\mathbf{R})$, in which each scatterer is initially a right rotator. By studying the origin of this behavior we are lead to introduce the concepts of an \emph{admissible configuration} and a \emph{blocking configuration}.

We show that if the particle moves on a lattice with an admissible configuration of scatterers then, between returns to its initial position, the particle's trajectory is a self-avoiding walk (cf. figures \ref{fig:3} and \ref{fig:16}). This self-avoiding property is novel in the sense that it is a consequence of the particle's deterministic dynamics. This is in contrast to the large majority of self-avoiding walks, which are generated via some random process. Moreover, many basic mathematical questions regarding self-avoiding walks are still open. Therefore, the results in this paper regarding this type of motion are potentially important to gaining a better mathematical understanding of self-avoiding walks. In particular, how deterministic dynamics can lead to self-avoiding motion.

\begin{figure}
    \begin{overpic}[scale=.2]{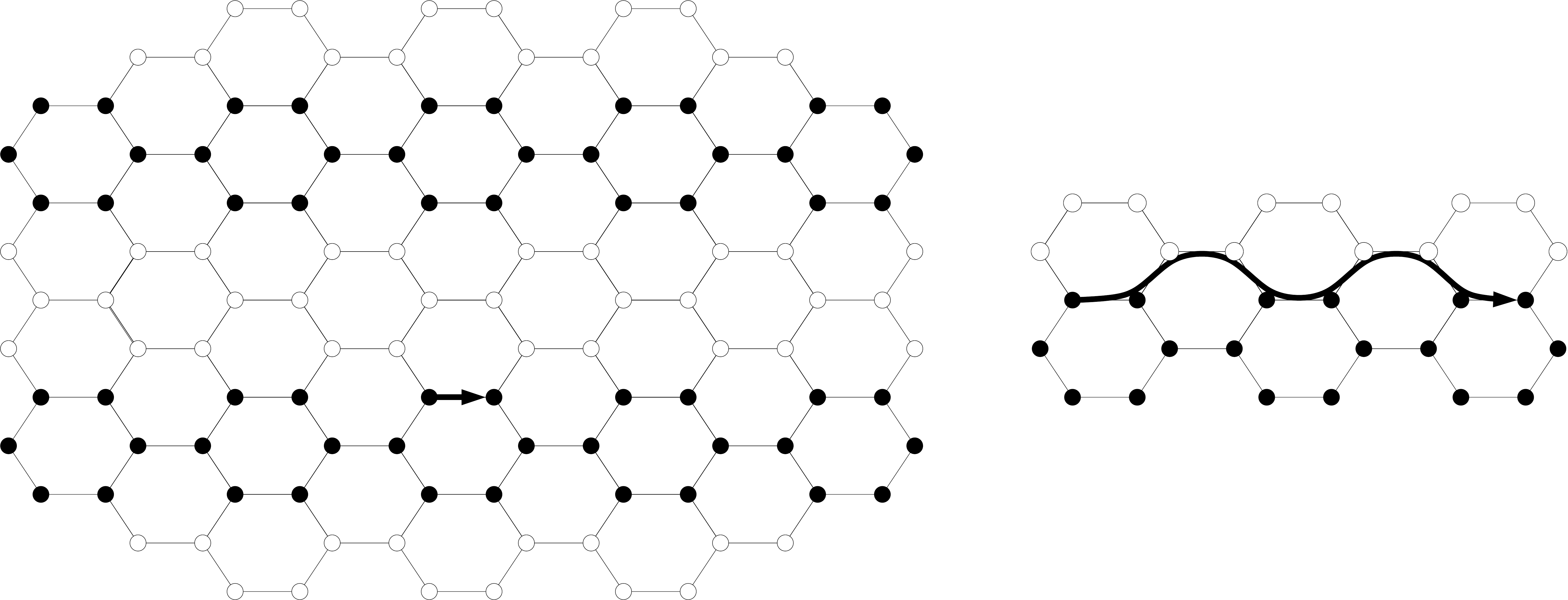}
    \put(73,9){(b) \emph{escape route}}
    \put(26,-3.5){(a) $B$}
    \put(25,12.25){$\mathbf{r}$}
    \put(28.5,13.75){$\mathbf{v}$}

    \put(65.75,18.5){$\mathbf{r}$}
    \put(69.25,20){$\mathbf{v}$}
    \end{overpic}
    \vspace{0.1in}
\caption{The admissible configuration $B$, shown in (a), and the particle's escape route for the indicated initial condition $I=(\mathbf{r},\mathbf{v})$, shown in (b).}\label{fig:18}
\end{figure}
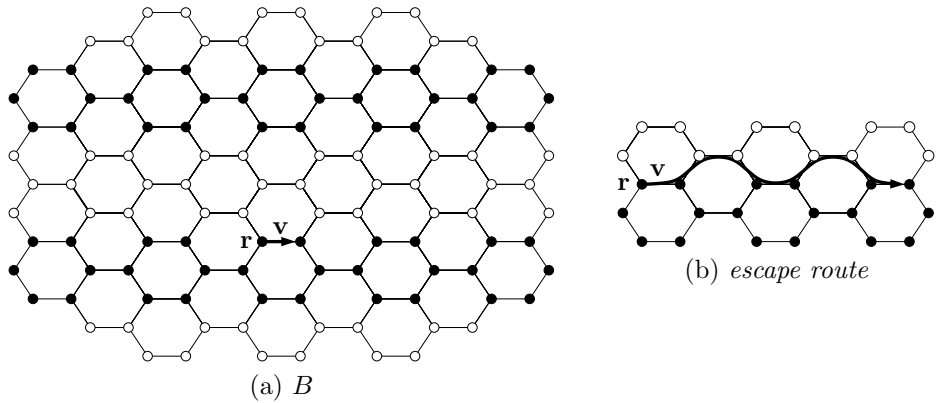

Aside from admissible configurations, the second type of configuration we consider in this paper is what we refer to as a blocking configuration. We show that if a particle moves on a lattice with an initial configuration that is both a blocking and an admissible configuration, the particle will experience an infinite number of returns to its initial position. Hence, the particle's motion can be decomposed into an infinite sequence of cycles through its initial position.

The LLG models we study, that have an admissible blocking configuration, not only exhibit this new type of motion, but also differ from those LLG models that have been previously studied in another way. The difference is that the models, which we study here, have only a single initial configuration, whereas those that have been previously considered have a large number of initial configurations. In order to describe the particle's dynamics in the models we consider, we introduce the notion of a \emph{time-averaged mean square displacement}. This allows us to characterize these model's dynamics in a way that is comparable to using the standard mean square displacement on an LLG model that has a large number of initial configurations.

Using this concept of a time-averaged mean square displacement, we also find that the particle in these models numerically exhibit time-averaged subdiffusion, suggesting that they have an unbounded trajectory. This leads us to introduce the idea of \emph{pulsating} dynamics, in which the particle repeatedly returns to its initial position but has an unbounded trajectory. This type of behavior is observed in a number of examples within the paper, specifically in those LLG models that have an admissible blocking configuration.

In fact, what we observe numerically for each LLG model that has both an admissible and a blocking configuration is that the particle has (a) a pulsating motion and (b) the particle exhibits time-averaged subdiffusion. Our conjecture is that both of these properties hold in general, so that in a LLG model with an admissible blocking configuration, the particle will always have an unbounded trajectory and a time-averaged mean square displacement $\bar{\triangle}(t)\approx ct^\alpha$, for some $c>0$ and $0<\alpha<1$.

Beyond these conjectures, there are other open questions regarding the geometric nature of the cycles generated in these systems, e.g. whether they have fractal-like properties as are often observed in self-avoiding walks. We note that an understanding of the geometry of these cycles and their dependence on the system's initial configuration could be important in using these LLGs to model the growth of crystals, polymers, etc.

\section{Appendix A}
In this appendix we give the lengths of the cycles $L(i)=\tau_i-\tau_{i-1}$ in the $(H;\mathbf{R})$ model for $1\leq i\leq 180$. These are shown in table \ref{table:length}, which is read from left to right and from top to bottom. Note that each cycle length in this list has the form $6+4n$ where $n$ is some nonnegative integer.

\begin{table}[ht]
\caption{Cycle Lengths in the $(H;\mathbf{R})$ Model} 
\centering 
\begin{tabular}{cccccccccccccccc} 
\hline\hline 
6&18&6&42&6&18&6&6&66&6&18&14&10&30&30\\ 
10&14&18&6&78&6&6&18&6&78&22&10&22&18&54\\
18&6&42&22&122&30&30&10&14&18&6&18&6&6&18\\
6&18&14&10&126&30&34&14&42&6&114&6&6&18&6\\
90&38&10&22&18&134&110&6&6&14&130&6&10&54&38\\
22&6&158&6&6&34&6&74&6&42&6&6&150&6&34\\
46&38&10&38&18&6&298&6&6&42&6&114&62&126&22\\
22&22&174&6&22&6&6&18&6&58&6&18&6&6&18\\
6&82&22&10&22&18&118&54&6&6&14&38&6&6&210\\
54&170&30&202&30&6&6&38&226&6&266&22&18&6&130\\
22&6&6&14&26&6&18&30&126&6&6&14&38&6&6\\
218&54&230&30&14&10&6&38&6&6&34&6&170&6&42\\
[1ex] 
\hline 
\end{tabular}
\label{table:length} 
\end{table}

\section{Appendix B}
Here, we give a proof of proposition \ref{thm:0}, found in section \ref{sec:4}. This theorem can be thought of as having two parts. The first part of the theorem states that if the particle in an LLG model diffuses, subdiffuses, superdiffuses, or propagates then the particle also exhibits time-averaged diffusion, subdiffusion, superdiffusion, or propagation, respectively. The second part of proposition \ref{thm:0} states that the converse of the first part of the theorem does not hold in general. We now give a proof of proposition \ref{thm:0}.

\begin{proof}
To prove part (a) of proposition \ref{thm:0}, suppose $\lim_{t\rightarrow\infty}\triangle(t)/t=c$ for some $c>0$. Then for any $\epsilon>0$ there exists a $T>0$ such that if $t>T$ then $c-\epsilon<\triangle(t)/t<c+\epsilon$. Hence, for $t>T$ we have $(c-\epsilon)t<\triangle(t)<(c+\epsilon)t$ implying
\begin{equation}
\frac{1}{t}\sum_{i=1}^T\triangle(i)+\frac{1}{t}\sum_{i=T+1}^t(c-\epsilon)i<\bar{\triangle}(t)< \frac{1}{t}\sum_{i=1}^T\triangle(i)+\frac{1}{t}\sum_{i=T+1}^t(c+\epsilon)i.
\end{equation}
Using the identity $\sum_{i=1}^n i=n(n+1)/2$ we have
\begin{equation}\label{eq:sumup}
\sum_{i=1}^T\frac{\triangle(i)}{t^2}+(c-\epsilon)\frac{(t+T)(t-T+1)}{2t^2}<\frac{\bar{\triangle}(t)}{t}< \sum_{i=1}^T\frac{\triangle(i)}{t^2}+(c+\epsilon)\frac{(t+T)(t-T+1)}{2t^2}
\end{equation}
for $t>T$. Using the inequalities in \eqref{eq:sumup}, it follows that
\begin{equation}
(c-\epsilon)/2<\lim_{t\rightarrow\infty}\bar{\triangle}(t)/t<(c+\epsilon)/2.
\end{equation}
Since $\epsilon$ is arbitrary, then $\lim_{t\rightarrow\infty}\bar{\triangle}(t)/t=c/2$ so that $\bar{\triangle}(t)\simeq t$. Therefore, if $\triangle(t)\simeq t$ then $\bar{\triangle}(t)\simeq t$ completing the proof of proposition \ref{thm:0} part (a).

Similarly, using the identity $\sum_{i=1}^n i^2=n(n+1)(2n+1)/6$ it follows that if $\lim_{t\rightarrow\infty}\triangle(t)/t^2=c$, where $c>0$, then $\lim_{t\rightarrow\infty}\bar{\triangle}(t)/t^2=c/3$.  Thus, if $\triangle(t)\simeq t^2$ then $\bar{\triangle}(t)\simeq t^2$, which verifies part (d) proposition \ref{thm:0}.

For part (b), suppose $\lim_{t\rightarrow\infty}c/\triangle(t)=0$ for some $c>0$. Then $\lim_{t\rightarrow\infty}\triangle(t)=\infty$, implying that for $M>0$ there is a $T>0$ such that $\triangle(t)>M$ for all $t>T$. As a consequence, for $t>T$,
$$\frac{1}{t}\sum_{i=1}^t \triangle(t)=\frac{1}{t}\sum_{i=1}^T\triangle(i)+\frac{1}{t}\sum_{i=T+1}^t\triangle(t)> \frac{1}{t}\sum_{i=1}^T\triangle(i)+\frac{1}{t}(t-T)M.$$
Thus, $\lim_{t\rightarrow\infty}\bar{\triangle}(t)>M$ implying $\lim_{t\rightarrow\infty}\bar{\triangle}(t)=\infty$ as $M$ is arbitrary. Hence, if $c\prec \triangle(t)$ then $c\prec \bar{\triangle}(t)$.

Now suppose that $\triangle(t)\prec t$. Then for any $\epsilon>0$ there is again a $T>0$ such that if $t>T$ then $\triangle(t)/t<\epsilon$. Hence, for $t>T$
\begin{equation}
\frac{\bar{\triangle}(t)}{t}<\frac{1}{t^2}\sum_{i=1}^{T}\triangle(i)+
\frac{1}{t}\sum_{i=T+1}^t\frac{\triangle(i)}{i}<\frac{1}{t^2}\sum_{i=1}^{T}\triangle(i)+
\frac{\epsilon(t-T+1)}{t}.
\end{equation}
Therefore, $\lim_{t\rightarrow\infty}\bar{\triangle}(t)/t<\epsilon$ for every $\epsilon>0$ implying that if $\triangle(t)\prec t$ then $\bar{\triangle}(t)\prec t$. This finishes the proof of part (b) of proposition \ref{thm:0}. The proof that part (c) holds is analogous to proving part $(b)$ and is therefore omitted.

We now consider the converse of parts (a)-(d) of proposition \ref{thm:0}. To do so, suppose that the mean square displacement $\triangle(t)$ of a particular LLG model is given by
$$\triangle(t)=
\begin{cases}
0 \ &\text{if} \ t=\text{even}\\
2t-1 \ &\text{if} \ t=\text{odd}
\end{cases}.$$
Since, $\lim_{t\rightarrow\infty}\triangle(t)/t$ does not exist, then in particular $\triangle(t)\not\simeq t$. However, as
$$\bar{\triangle}(t)=
\begin{cases}
t/2 \ &\text{if} \ t=\text{even}\\
(t+1)/2 \ &\text{if} \ t=\text{odd}
\end{cases}$$
then $\lim_{t\rightarrow\infty}\bar{\triangle}(t)/t=1/2$ implying $\bar{\triangle}(t)\simeq t$. Hence, the converse of part (a) does not hold. Similar examples can be found demonstrating that the converse of parts (b)-(d) do not hold in general.
\end{proof}

\end{spacing}

\end{document}